\documentclass[twoside,journal]{IEEEtran}
\IEEEoverridecommandlockouts
\renewcommand\IEEEkeywordsname{Index Terms}
\ifCLASSINFOpdf
\else
\fi
\usepackage{xcolor,soul,framed} 
\colorlet{shadecolor}{yellow}
\usepackage[pdftex]{graphicx}
\graphicspath{{../pdf/}{../jpeg/}}
\DeclareGraphicsExtensions{.pdf,.jpeg,.png}
\usepackage{cite}
\usepackage{amsmath,amssymb,amsfonts}
\usepackage{algpseudocode}
\PassOptionsToPackage{ruled, linesnumbered}{algorithm2e}
\usepackage{algorithm2e}
\SetKw{Break}{break out of the loop}
\usepackage{graphicx}
\usepackage{textcomp}
\usepackage{comment}
\usepackage{units}
\usepackage{url}
\usepackage{geometry}
\usepackage{bbding}
\usepackage[caption=false,font=footnotesize]{subfig}
\usepackage{mathrsfs}
\usepackage{xcolor}
\definecolor{myblue}{rgb}{0.0, 0.5, 1.0}
\definecolor{myred}{rgb}{1.0, 0.13, 0.32}
\definecolor{mygreen}{rgb}{0.31, 0.68, 0.07}
\definecolor{grn}{rgb}{0.0, 0.5, 0.0}
\usepackage[pdftex]{graphicx}
\DeclareGraphicsExtensions{.pdf,.jpeg,.png}
\def\BibTeX{{\rm B\kern-.05em{\sc i\kern-.025em b}\kern-.08em
    T\kern-.1667em\lower.7ex\hbox{E}\kern-.125emX}}

\newtheorem{lemma}{Lemma}

\newtheorem{remark}{Remark}

\usepackage{pgfplots}
\usepackage{pgfplotstable}
\usepackage{tikz}
\usetikzlibrary{calc}

\usepackage{hyperref}
\hypersetup{
	colorlinks=true,
	linkcolor=blue,
	filecolor=magenta,      
	urlcolor=blue,
	citecolor=blue,
}
\usepackage{todonotes}
\usepackage{balance}
\allowdisplaybreaks

\begin{document}
\bstctlcite{IEEEexample:BSTcontrol}

   \title{\huge  Energy Efficiency in Microwave Linear Analog Computer (MiLAC)-Enabled Communications} 
    	\newgeometry {top=25.4mm,left=19.1mm, right= 19.1mm,bottom =19.1mm}%
\author{Ahmed Magbool,~\IEEEmembership{Member,~IEEE,} and Marco Di Renzo,~\IEEEmembership{Fellow,~IEEE}\thanks{Ahmed Magbool and Marco Di Renzo are with the Centre for Telecommunications Research, Department of Engineering, King's College London, WC2R 2LS London, United Kingdom (e-mail: ahmed.1.magbool@kcl.ac.uk; marco.di\_renzo@kcl.ac.uk). Marco Di Renzo is also with CNRS and CentraleSup\'elec, Institute of Electronics and Digital Technologies (IETR), Avenue de la Boulaie, 35576 Cesson-S\'evign\'e, France (e-mail: marco.direnzo@centralesupelec.fr). \par
This work was supported in part by the European Research Council (ERC) under the European Union’s Horizon Europe Programme WePhICom (agreement number 101225119), as well as by the European Union through the Horizon Europe project COVER under grant agreement number 101086228, the Horizon Europe project UNITE under grant agreement number 101129618, the Horizon Europe project INSTINCT under grant agreement number 101139161, and the Horizon Europe project TWIN6G under grant agreement number 101182794, as well as by the Agence Nationale de la Recherche (ANR) through the France 2030 project ANR-PEPR Networks of the Future under grant agreement NF-YACARI 22-PEFT-0005, and by the CHIST-ERA project PASSIONATE under grant agreements CHIST-ERA-22-WAI-04 and ANR-23-CHR4-0003-01. Also, the work of M. Di Renzo was supported in part by the Engineering and Physical Sciences Research Council (EPSRC), part of UK Research and Innovation, and the UK Department of Science, Innovation and Technology through the CHEDDAR Telecom Hub under grant EP/Y037421/1, through the HASC Telecom Hub under grant EP/Y037197/1, and through the TITAN Telecom Hub under grant EP/Y037243/1.}\vspace{-1.1cm}}

\maketitle
\begin{abstract}
Microwave linear analog computers (MiLACs) have emerged as a promising architecture for energy-efficient wireless communications by shifting signal processing from the digital to the analog domain using tunable impedance networks. Several MiLAC architectures have recently been proposed, including the single-layer MiLAC (SLM), two-layer MiLAC (TLM), and hybrid digital-MiLAC (HDM). In this paper, we investigate the energy efficiency (EE) of these architectures and establish a unified framework for their modeling and optimization. Specifically, we formulate EE maximization problems for the SLM, TLM, and HDM architectures under transmit power, user rate, and architecture-specific constraints, and develop a dimensionality reduction technique with successive convex approximation (SCA)-based algorithms to efficiently solve the resulting non-convex problems. We further derive a computationally efficient solution for EE maximization based on a search over only $(K+1)$ closed-form candidate solutions, where $K$ is the number of users, and analyze the asymptotic EE in the large-antenna regime under negligible quantization noise. Our analysis shows that the EE of the SLM, TLM, and HDM architectures scales as $\ln(N)/N^2$, whereas conventional digital beamforming (DBF) scales as $\ln(N)/N$, where $N$ denotes the number of transmit antennas. Despite these different scaling laws, the faster EE decay of MiLAC-based architectures becomes relevant only at very large antenna dimensions, typically involving thousands of antennas, while they maintain superior EE over practically relevant finite antenna regimes. Finally, we derive an approximation of the number of antennas required to achieve the maximum EE for each architecture and the corresponding EE. Simulation results show that, over practically relevant antenna regimes, MiLAC-based architectures achieve substantially higher EE than conventional fully digital and hybrid analog-digital beamforming.

\vspace{-0.1cm}
\end{abstract}

\IEEEpeerreviewmaketitle
\vspace{-0.3cm}
\section{Introduction} \label{sec:intro}
\vspace{-0.15cm}

The increasing demand for wireless connectivity has motivated the development of advanced network architectures with higher capacity and spectral efficiency (SE). To this end, massive multiple-input multiple-output (MIMO) has emerged as a key enabling technology, equipping base stations (BSs) with large-scale antenna arrays to simultaneously serve multiple users over the same time-frequency resources and providing significant gains in SE, spatial multiplexing, and link reliability. However, conventional fully digital beamforming (DBF) requires a dedicated radio-frequency (RF) chain for each antenna, leading to substantial hardware power consumption and making energy efficiency (EE) a critical bottleneck for practical massive MIMO deployment~\cite{2015_Bjornson}.

To address the power consumption of large-scale antenna arrays, several energy-efficient MIMO architectures have been proposed~\cite{2025_You}, including hybrid analog-digital beamforming (HAD)~\cite{2016_Sohrabi}, low-resolution ADC/DAC architectures~\cite{2017_Abbas}, intelligent metasurfaces~\cite{2025_Magbool1,2026_Magbool}, and movable antenna systems~\cite{2026_Zhu}. More recently, microwave linear analog computers (MiLACs) have emerged as a promising architecture for energy-efficient beamforming~\cite{2025_Nerini}. A MiLAC is an analog computing architecture that performs linear signal processing through programmable microwave circuits, where tunable impedance networks implement configurable linear transformations for beamforming~\cite{2025_Nerini1}. By shifting signal processing from the digital to the analog domain, MiLACs can reduce reliance on power-intensive digital processing and RF chains, thereby enabling low-complexity, low-power beamforming for large-scale antenna arrays~\cite{2025_Nerini,2025_Nerini1}.

The use of MiLACs for computationally efficient mathematical operations was introduced in~\cite{2025_Nerini}, where they were shown to implement signal processing operations with significantly reduced complexity. This was later extended to wireless communications by applying MiLACs to large-scale MIMO beamforming~\cite{2025_Nerini1}. The work of~\cite{2026_Nerini2,2026_Nerini3} further demonstrated that MiLACs can achieve the same capacity as DBF in point-to-point MIMO systems while substantially reducing power consumption. However, practical hardware constraints prevent MiLAC beamforming from fully matching DBF performance in multi-user multiple-input single-output (MU-MISO) systems, creating a fundamental tradeoff between performance and power savings~\cite{2026_Wu}. Recent works have also explored MiLAC architectures with reduced circuit complexity~\cite{2026_Nerini4,2026_Zhang1} and extended their applications to simultaneous active and passive beamforming~\cite{2026_Nerini5}, wideband systems~\cite{2026_Peng}, and sensing applications~\cite{2026_Liu,2026_Zhang2}.

The conventional MiLAC architecture, referred to as the single-layer MiLAC (SLM), consists of an active front end with a power allocation stage and RF chains, followed by a passive MiLAC network. Two enhanced architectures have been proposed: the hybrid digital-MiLAC (HDM), which incorporates a small DBF module in the active front end~\cite{2026_Wu}, and the two-layer MiLAC (TLM), which adds a MiLAC layer before power allocation~\cite{2026_Zhou}. Both are shown to achieve optimal MU-MISO beamforming performance by matching the DBF performance~\cite{2026_Zhou}.

While a key motivation for MiLACs is their potential to improve the EE of large-scale wireless systems, comprehensive studies of the EE of MiLAC-based architectures remain limited. The authors of~\cite{2026_Zhang} developed an EE maximization framework for MiLAC-aided MU-MISO systems and characterized the EE-SE tradeoff. However, the framework does not consider minimum rate constraints, which are critical for satisfying quality-of-service (QoS) requirements, particularly in millimeter-wave (mmWave) systems with severe path loss and channel disparities~\cite{2025_Magbool}. Moreover, the study considers only the conventional single-layer MiLAC (SLM) architecture and does not investigate the EE performance or EE-SE tradeoffs of the TLM and HDM architectures, nor does it characterize their asymptotic EE scaling behavior.

Motivated by this, we investigate the EE of SLM, TLM, and HDM architectures and compare them with the DBF and hybrid analog-digital (HAD) architectures. The main contributions of this paper are summarized as follows:
\begin{itemize}
    \item We develop a unified system model for the SLM, TLM, and HDM architectures, covering their transmit, receive, and power consumption models, while accounting for low-resolution RF chains through the additive quantization noise model (AQNM).

    \item We formulate EE maximization problems for the considered architectures, incorporating the maximum transmit power budget, minimum user rate requirements, and architecture-specific constraints.

    \item Unlike~\cite{2026_Zhang}, which directly employs a minimum mean squared error (MMSE)-based solution, the inclusion of minimum user rate constraints makes our optimization problem more challenging. We therefore develop efficient solution methods based on dimensionality reduction and successive convex approximation (SCA).

    \item We derive low-complexity EE solutions using large-antenna approximations under negligible quantization noise, requiring a search over only $(K+1)$ closed-form candidate solutions, where $K$ denotes the number of users. These solutions are highly accurate for orthogonal channels and provide useful insights into general system setups. We further show that the EE of the SLM, TLM, and HDM architectures scales as $\ln(N)/N^2$, compared with $\ln(N)/N$ for DBF, where $N$ denotes the number of transmit antennas, indicating a faster EE decay for MiLAC-based architectures in the asymptotically large-antenna regime, despite their superior EE over practically relevant finite antenna regimes. Finally, we derive an approximate optimal number of antennas that maximizes EE and its corresponding value.

    \item We present extensive simulations showing that MiLAC-based architectures achieve higher EE than conventional DBF and HAD systems over practical antenna regimes. The TLM achieves the highest EE by realizing DBF-equivalent performance with lower hardware power consumption, while HDM becomes more efficient with very low-resolution RF chains by mitigating quantization noise through digital preprocessing.
\end{itemize}

The rest of the paper is organized as follows. Section~\ref{sec:sys_model} introduces the signal models for the SLM, TLM, and HDM architectures, as well as their power consumption models. Section~\ref{sec:prob_for} formulates the EE maximization problems, while Section~\ref{sec:sol} describes the proposed solution methods. Section~\ref{sec:Asy_ana} presents the asymptotic analysis. Section~\ref{sec:sim} discusses the numerical results, and Section~\ref{sec:conc} concludes the paper.

\textit{Notation:} Bold lowercase and uppercase letters denote vectors and matrices, respectively. $\Re\{\cdot\}$, $|\cdot|$, and $(\cdot)^*$ denote the real part, magnitude, and complex conjugate, while $(\cdot)^{\mathrm{T}}$ and $(\cdot)^{\mathrm{H}}$ denote the transpose and conjugate transpose. $\|\cdot\|_2$ and $\|\cdot\|_{\mathrm{F}}$ denote the Euclidean and Frobenius norms, respectively. $\mathbf{I}_a$ and $\mathbf{0}_{a\times b}$ denote the $a\times a$ identity matrix and the $a\times b$ zero matrix, respectively. $\mathrm{diag}(\mathbf{A})$ denotes the diagonal matrix formed by the
diagonal entries of $\mathbf{A}$, i.e., the matrix obtained by setting all
off-diagonal entries of $\mathbf{A}$ to zero, while
$\mathrm{diag}(a_1,\dots,a_N)$ denotes the diagonal matrix with diagonal
entries $a_1,\dots,a_N$. $[\mathbf{A}]_{a:b,c:d}$ denotes the submatrix of $\mathbf{A}$ formed by rows $a$--$b$ and columns $c$--$d$. $\mathbb{C}$ denotes the set of complex numbers, $j=\sqrt{-1}$, and $\mathbb{E}\{\cdot\}$ the expectation operator. $\mathcal{CN}(\mathbf{a},\mathbf{B})$ denotes a complex Gaussian distribution with mean $\mathbf{a}$ and covariance matrix $\mathbf{B}$. The notation $\xrightarrow[a\rightarrow\infty]{\mathrm{a.s.}}$ denotes almost sure convergence as $a\rightarrow\infty$. Finally, $\mathcal{O}(\cdot)$ denotes the big-O computational complexity.

\vspace{-0.3cm}
\section{System Model} \label{sec:sys_model}

\begin{figure}
         \centering 
         \includegraphics[width=0.95\columnwidth]{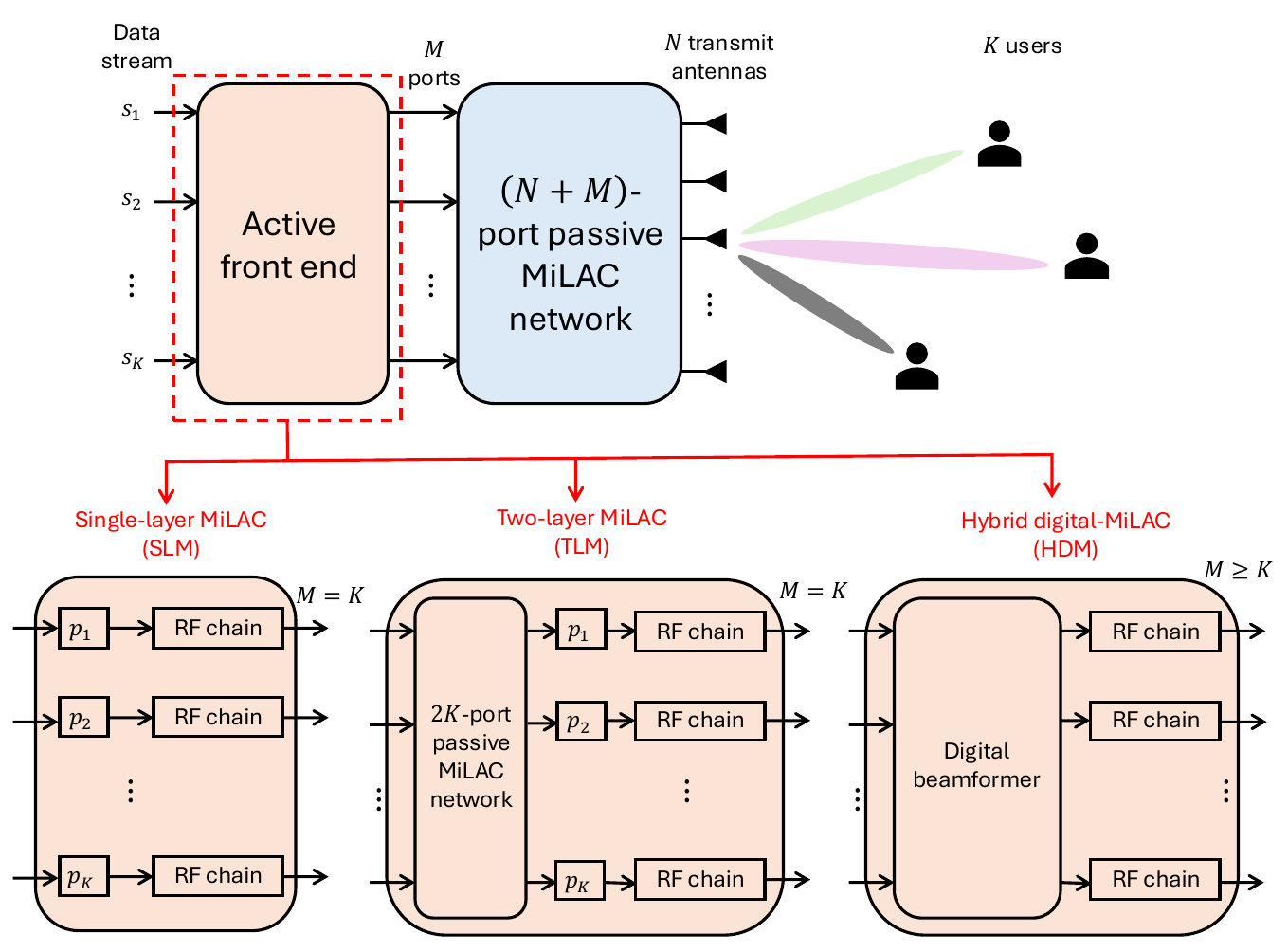}
         \vspace{-0.2cm}
        \caption{Considered system model consisting of an active front end and a passive MiLAC network for the SLM, TLM, and HDM architectures.}
        \label{fig:sys_model}
\end{figure}

We consider a downlink communication system consisting of a BS equipped with $N$ transmit antennas and serving $K$ single-antenna users. As depicted in Fig.~\ref{fig:sys_model}, the system employs an active front end followed by a passive MiLAC network. Our goal is to study the EE of three MiLAC-based systems: (i) the SLM system, whose front end comprises power allocation followed by $M=K$ RF chains; (ii) the TLM system, whose front end consists of a smaller $2K$-port MiLAC network followed by $M=K$ RF chains, and (iii) the HDM system, whose front end consists of a digital beamformer followed by $M \geq K$ RF chains. 

In the following, we first present the transmit and receive signal models, followed by the power consumption models for the three architectures.

\vspace{-0.4cm}

{
\subsection{Transmit and Receive Signal Models}

We first develop a general transmit and receive signal model that captures a general front end. Let $g\in\{\mathrm{SLM},\mathrm{TLM},\mathrm{HDM}\}$ denote the considered architecture. We assume that each RF chain is equipped with a $b$-bit DAC and employ the AQNM to characterize the quantization noise~\cite{2026_Zhang}. Accordingly, the output signal of the active front end for architecture $g$ is expressed as
\begin{equation}
    \mathbf{c}^{g}
    =
    \varpi \mathbf{A}^{g}\mathbf{s}
    +
    \mathbf{n}_{\mathrm{q}}^{g},
\end{equation}
where $\mathbf{s}=[s_1,\dots,s_K]^{\mathrm{T}}\in\mathbb{C}^{K\times1}$ denotes the data vector, with $s_k$ representing the data symbol intended for the $k$-th user, and $\mathbf{s}\sim\mathcal{CN}(\mathbf{0}_{K\times1},\mathbf{I}_K)$. Moreover, $\mathbf{A}^{g}\in\mathbb{C}^{M\times K}$ denotes the active beamforming matrix of architecture $g$, whose structure depends on the considered architecture, and \(\varpi\) denotes the quantization gain factor, whose value depends on the DAC resolution and is reported in~\cite{2019_Dai}. Also, $\mathbf{n}_{\mathrm{q}}^{g}\sim\mathcal{CN}(\mathbf{0},\mathbf{R}_{\mathrm{q}}(\mathbf{A}^g))$ is the quantization noise where $\mathbf{R}_{\mathrm{q}}(\mathbf{A}^{g})\triangleq\varpi(1-\varpi)\mathrm{diag}\left(\mathbf{A}^{g}(\mathbf{A}^{g})^{\mathrm{H}}\right)$ denotes the quantization noise covariance matrix~\cite{2026_Zhang}. 

The output signal of the active front end is then processed by the MiLAC network, resulting in the transmitted signal
\begin{equation}
    \mathbf{x}^{g}
    =
    \mathbf{F}\mathbf{c}^{g}
    =
    \varpi\mathbf{F}\mathbf{A}^{g}\mathbf{s}
    +
    \mathbf{F}\mathbf{n}_{\mathrm{q}}^{g},
\end{equation}
where $\mathbf{F} = \frac{1}{2}
    \left[
    \mathbf{\Phi}
    \right]_{M+1:M+N,\;1:M}\in\mathbb{C}^{N\times M}$ denotes the MiLAC beamforming matrix, which maps the $M$ input ports connected to the RF chains to the $N$ output ports connected to the transmit antennas, where $\mathbf{\Phi}$ denotes the scattering matrix of the passive MiLAC network. The relationship between the scattering matrix and the network admittance matrix is given by $\mathbf{\Phi}
    =
    \left(\mathbf{I}+Z_0\mathbf{Y}\right)^{-1}
    \left(\mathbf{I}-Z_0\mathbf{Y}\right)$, where $Z_0$ denotes the reference impedance. For a reciprocal and lossless MiLAC network, the scattering matrix should satisfy $\mathbf{\Phi}
    =
    \mathbf{\Phi}^{\mathrm T}$ and $ 
    \mathbf{\Phi}^{\mathrm H}\mathbf{\Phi}
    =
    \mathbf{I}_{N+M}$\cite{2026_Wu}.

The received signal at the $k$-th user is given by
\begin{equation}
    y_k^{g}
    =
    \mathbf{h}_k^{\mathrm H}\mathbf{x}^{g}
    +
    n_k,
\end{equation}
where $\mathbf{h}_k^{\mathrm H}\in\mathbb{C}^{1\times N}$ denotes the channel between the transmit antenna array and the $k$-th user, and $n_k\sim\mathcal{CN}(0,\sigma_k^2)$ denotes the additive white Gaussian noise (AWGN).

Accordingly, the signal-to-interference-plus-noise ratio
 (SINR) of the $k$-th user is expressed as
\begin{equation}
    \gamma_k^{g}
    (\mathbf A^{g},\mathbf F)
    =
    \frac{
    \varpi^2
    \left|
    \mathbf h_k^{\mathrm H}
    \mathbf F\mathbf a_k^{g}
    \right|^2
    }
    {
    \varpi^2
    \sum_{i\in\mathcal K\setminus\{k\}}
    \left|
    \mathbf h_k^{\mathrm H}
    \mathbf F\mathbf a_i^{g}
    \right|^2
    +
    \varsigma_k(\mathbf A^{g},\mathbf F)
    +
    \sigma_k^2
    },
\end{equation}
where $\mathbf a_i^{g}$ denotes the $i$-th column of $\mathbf A^{g}$, $\varsigma_k(\mathbf A^{g},\mathbf F)
    =
    \varpi(1-\varpi)
    \mathbf h_k^{\mathrm H}
    \mathbf F
    \mathrm{diag}
    \left(
    \mathbf A^{g}
    (\mathbf A^{g})^{\mathrm H}
    \right)
    \mathbf F^{\mathrm H}
    \mathbf h_k$ is the received quantization noise power at user $k$~\cite{2026_Zhang}, and $\mathcal K=\{1,\dots,K\}$.

Therefore, the achievable sum rate is expressed as
\begin{equation}
    R_{\mathrm{sum}}^{g}
    (\mathbf A^{g},\mathbf F)
    =
    \sum_{k\in\mathcal K}
    \underbrace{\log_2
    \left(
    1+
    \gamma_k^{g}
    (\mathbf A^{g},\mathbf F)
    \right)}_{R^g_k}.
\end{equation}
Next, we present $\mathbf{A}^g$ for the three architectures.

\subsubsection{SLM} For the SLM architecture, the active front end consists of power allocation followed by $K$ RF chains, where each RF chain directly feeds one input port of the single-layer MiLAC network. Hence, the active beamforming matrix $\mathbf A^{\mathrm{SLM}}
    =
    \mathbf P$, where $\mathbf P
    \triangleq
    \mathrm{diag}
    (\sqrt{p_1},\dots,\sqrt{p_K})$ is the power allocation matrix.

\subsubsection{TLM} For the TLM architecture, the active front end consists of a first MiLAC layer followed by power allocation and $K$ RF chains. The first MiLAC layer performs an analog transformation before the DACs, while the second MiLAC layer is represented by the common MiLAC matrix $\mathbf F$. Therefore, the active beamforming matrix is given by $\mathbf A^{\mathrm{TLM}}
    =
    \mathbf P \boldsymbol{\Xi},$ where $\boldsymbol{\Xi}$ denotes the beamforming matrix of the first MiLAC layer.

\subsubsection{HDM} For the HDM architecture, the active front end consists of a digital precoder followed by $M\geq K$ RF chains. The digital precoder provides additional degrees of control before the MiLAC transformation, and hence $ \mathbf A^{\mathrm{HDM}}
    =
    \mathbf D$, where $\mathbf D\in\mathbb C^{M\times K}$ denotes the digital precoding matrix.}

\vspace{-0.3cm}

\subsection{Power Consumption Model}

The total power consumption of a communication system consists of the power consumed by the baseband processing $P_{\mathrm{BB}}$, as well as the static power required by the underlying hardware components $P_{\mathrm{cir}}$~\cite{2016_Buzzi}.

The baseband power consumption of architecture $g$ is
\begin{equation}
    P_{\mathrm{BB}}(\mathbf{A}^{g})
    =
    \frac{\varpi}{\rho}
    \left\|
    \mathbf{A}^{g}
    \right\|_{\mathrm{F}}^2,
\end{equation}
where $\rho$ denotes the power amplifier efficiency.

The circuit power consumption depends on the considered architecture, as different MiLAC configurations require different numbers of RF chains and tunable microwave components. The impact of these parameters on the overall energy efficiency is further discussed in Section~\ref{sec:sim}. The detailed power consumption parameters of the considered architectures are summarized in Table~\ref{tab:2}. For the SLM architecture, the circuit power consumption is given by
\begin{equation}
\begin{split}
P_{\mathrm{cir}}^{\mathrm{SLM}}
=&\,
P_{\mathrm{LO}}
+
K\!\Big(
P_{\mathrm{H}}
+
2P_{\mathrm{M}}
+
2P_{\mathrm{VGA}}
+ 2P_{\mathrm{LP}} \\
&
+
2P_{\mathrm{ADC}}
\Big)
+
\frac{(K+N)(K+N+1)}{2}
P_{\mathrm{IT}}.
\end{split}
\label{eq:SLM_p_cir}
\end{equation}
For the TLM architecture, an additional MiLAC layer is introduced, requiring extra tunable microwave components. Consequently, the circuit power consumption is expressed as
\begin{equation}
\begin{split}
P_{\mathrm{cir}}^{\mathrm{TLM}}
=&\,
P_{\mathrm{LO}}
+
K\!\left(
P_{\mathrm{H}}
+
2P_{\mathrm{M}}
+
2P_{\mathrm{VGA}}
+
2P_{\mathrm{LP}}
+
2P_{\mathrm{ADC}}
\right)\\
&
+
\left(
\frac{(K+N)(K+N+1)}{2}
+
K(2K+1)
\right)
P_{\mathrm{IT}}.
\end{split}
\end{equation}
Finally, the HDM architecture combines a DBF stage with a MiLAC layer. Therefore, its circuit power consumption is given by
\begin{equation}
\begin{split}
P_{\mathrm{cir}}^{\mathrm{HDM}}
=&
P_{\mathrm{LO}}
+
M\!\Big(
P_{\mathrm{H}}
+
2P_{\mathrm{M}}
+
2P_{\mathrm{VGA}}
+
2P_{\mathrm{LP}}\\
&
+
2P_{\mathrm{ADC}}
\Big) +
\frac{(M+N)(M+N+1)}{2}
P_{\mathrm{IT}}.
\end{split}
\end{equation}

Thus, the total power consumption of architecture $g$ is
\begin{equation}
P_{\mathrm{tot}}^{g}(\mathbf{A}^{g})
=
P_{\mathrm{BB}}^{g}(\mathbf{A}^{g})
+
P_{\mathrm{cir}}^{g},
\end{equation}
and the corresponding EE is expressed as
\begin{equation}
    \eta^{g}(\mathbf{A}^{g},\mathbf{F})
    =
    \frac{
    R_{\mathrm{sum}}^{g}(\mathbf{A}^{g},\mathbf{F})
    }
    {
    P_{\mathrm{tot}}^{g}(\mathbf{A}^{g})
    } .
\end{equation}

\vspace{-0.3cm}

\section{Problem Formulation} \label{sec:prob_for}
 \subsubsection{SLM} For the SLM architecture, the effective beamforming matrix can be represented directly by $\mathbf{W}\triangleq
    \mathbf{F}
    \mathbf{P}$, making it analogous to the DBF matrix. To account for the reciprocal and lossless constraints of the MiLAC network, $\mathbf{W}$ must satisfy $\mathbf{W}^{\mathrm{H}}\mathbf{W}
    \preccurlyeq
    \mathrm{diag}(p_1,\dots,p_K),$ as shown in~\cite{2026_Wu}.

Moreover, the received quantization noise power can be written as
\begin{equation}
\begin{aligned}
    & \varsigma_k(\mathbf{W})
    =
    \varpi(1-\varpi)
    \mathbf{h}_k^{\mathrm H}
    \mathbf{F}
    \text{diag}(\mathbf{P}\mathbf{P})
    \mathbf{F}^{\mathrm H}
    \mathbf{h}_k
    \! = \! \varpi (1 \\
    &  - \varpi)
    \mathbf{h}_k^{\mathrm H}
    \mathbf{W}\mathbf{W}^{\mathrm H}
    \mathbf{h}_k  \! = \!
    \varpi(1-\varpi)
    \left\|
    \mathbf{W}^{\mathrm H}
    \mathbf{h}_k
    \right\|_2^2.
\end{aligned}
\end{equation}
Accordingly, after applying the transformation $\mathbf{W}=\mathbf{F}\mathbf{P}$, the SINR, user rate, sum rate, and EE are expressed as functions of the effective beamforming matrix $\mathbf{W}$.

The baseband power consumption can also be expressed in terms of the power allocation matrix $\mathbf{P}$ as follows:
\begin{equation}
    P_{\mathrm{BB}}^{\mathrm{SLM}}(\mathbf{P})
    =
    \frac{\varpi}{\rho} \sum_{k \in \mathcal{K}} p_k.
\end{equation}

 The EE maximization problem can be formulated as follows:
\begin{subequations}
\label{eq:PF_SLM}
\begin{align}
(\mathcal{P}_1):\quad
\underset{\mathbf{W},\mathbf{P}}{\max}\quad
&
\eta^{\mathrm{SLM}}(\mathbf{W},\mathbf{P})\\
\text{s.t.}\quad
&
\varpi\sum_{k\in\mathcal K}p_k
\le
P_{\mathrm T}, \label{eq:PF_SLM_C1}\\
&
\mathbf{W}^{\mathrm H}\mathbf{W}
\preccurlyeq
\mathrm{diag}(p_1,\ldots,p_K), \label{eq:PF_SLM_C2}\\
&
R_k^{\mathrm{SLM}}(\mathbf{W})
\ge
r,\quad
\forall k\in\mathcal K \label{eq:PF_SLM_C3}.
\end{align}
\end{subequations}
Here,~\eqref{eq:PF_SLM_C1} imposes the transmit power budget, where $P_{\mathrm{T}}$ denotes the transmit power budget,~\eqref{eq:PF_SLM_C2} characterizes the feasible set of the effective beamforming matrix, while~\eqref{eq:PF_SLM_C3} ensures that each user satisfies the minimum rate requirement, where $r$ denotes the minimum rate threshold.
\begin{table*}[t]
\centering
\caption{Circuit power consumption models for the SLM, TLM, and HDM architectures. In the table, $f_{\mathrm{s}}$ denotes the DAC sampling frequency. The typical power consumption values of the individual components are based on~\cite{2018_Ribeiro,2026_Zhang} and~\cite{2017_Roth}.}
\vspace{-0.2cm}
\begin{tabular}{|l|l|c|c|c|c|}
\hline

\textbf{Component} & \textbf{Typical Value} & 
\multicolumn{3}{c|}{\textbf{Quantity / Count}} \\
\cline{3-5}
& & \textbf{SLM} & \textbf{TLM} & \textbf{HDM} \\
\hline

$P_{\mathrm{LO}}$: Local oscillator
& 22.5 mW
& $1$
& $1$
& $1$ \\
\hline

$P_{\mathrm{H}}$: Hybrid coupler
& 3 mW
& $K$
& $K$
& $M$ \\
\hline

$P_{\mathrm{M}}$: Mixer
& 0.3 mW
& $2K$
& $2K$
& $2M$ \\
\hline

$P_{\mathrm{VGA}}$: Variable gain amplifier
& 0.8–2 mW
& $2K$
& $2K$
& $2M$ \\
\hline

$P_{\mathrm{LP}}$: Low-pass filter
& 22.5 mW
& $2K$
& $2K$
& $2M$ \\
\hline

$P_{\mathrm{ADC}}$: ADC/DAC
& $1.5 \times 10^{-5}2^b + 9 \times 10^{-12} f_\mathrm{s}b$ 
& $2K$
& $2K$
& $2M$ \\
\hline

$P_{\mathrm{IT}}$: Impedance tuning
& $8\,\mu\mathrm{W}$
& $\frac{(K+N)(K+N+1)}{2}$
& $\frac{(K+N)(K+N+1)}{2}+K(2K+1)$
& $\frac{(M+N)(M+N+1)}{2}$ \\
\hline

\end{tabular}
\label{tab:2}
\vspace{-0.7cm}
\end{table*}

\subsubsection{TLM} The effective beamforming matrix for the TLM architecture is defined as $  \mathbf{W}
    \triangleq
    \mathbf{F}
    \mathbf{P}
    \boldsymbol{\Xi}.$ Unlike the SLM architecture, the TLM architecture can eliminate the constraint~\eqref{eq:PF_SLM_C2}, thereby enabling the realization of a fully digital beamformer as shown in~\cite{2026_Zhou}. Specifically, for any desired DBF matrix $\mathbf{W}_{\mathrm{d}}$, its singular value decomposition (SVD) is given by $\mathbf{W}_{\mathrm{d}}
    =
    \mathbf{U}
    \mathbf{\Sigma}
    \mathbf{V}^{\mathrm{H}},$ where $\mathbf{U}$ and $\mathbf{V}$ contain the left and right singular vectors, respectively, and $\mathbf{\Sigma}$ is a diagonal matrix containing the singular values. Accordingly, the two MiLAC layers and the power allocation matrix can be selected as $ \mathbf{F}
    =
    \frac{1}{2}\mathbf{U},
    \mathbf{P}
    =
    4\mathbf{\Sigma}$ and $
    \boldsymbol{\Xi}
    =
    \frac{1}{2}\mathbf{V}^{\mathrm{H}}$~\cite{2026_Zhou}.
    
The received quantization noise power at the $k$-th user can be expressed as
\begin{equation}
\begin{aligned}
    &\varsigma_k(\mathbf{W})
    =
    \varpi(1-\varpi)
    \mathbf{h}_{k}^{\mathrm H}
    \mathbf{F}
    \mathrm{diag}(\mathbf{P} \boldsymbol{\Xi}\boldsymbol{\Xi}^\mathrm{H} \mathbf{P})
    \mathbf{F}^{\mathrm H}
    \mathbf{h}_{k}\\
   \! & \! =\!
    \varpi(1-\varpi)
    \mathbf{h}_{k}^{\mathrm H}
    \mathbf{W}
    \mathbf{W}^{\mathrm H}
    \mathbf{h}_{k}=
    \varpi(1-\varpi)
    \left\|
    \mathbf{W}^{\mathrm H}
    \mathbf{h}_{k}
    \right\|_2^2 .
\end{aligned}
\end{equation}

The baseband power consumption can also be expressed in terms of the effective beamforming matrix $\mathbf{W}$ as follows:
\begin{equation}
    P_{\mathrm{BB}}^{\mathrm{TLM}}(\mathbf{W})
    =
    \frac{\varpi}{\rho}
    \left\|
    \mathbf{W}
    \right\|_{\mathrm{F}}^{2}.
\end{equation}

Accordingly, after applying the transformation $\mathbf{W}=\mathbf{F}\mathbf{P}\boldsymbol{\Xi}$, the EE maximization problem for the TLM architecture is formulated as
\begin{subequations}
\label{eq:PF_TLM}
\begin{align}
(\mathcal{P}_2):\quad
\underset{\mathbf{W}}{\max}\quad
&
\eta^{\mathrm{TLM}}(\mathbf{W})\\
\text{s.t.}\quad
&
\varpi
\left\|
\mathbf{W}
\right\|_{\mathrm F}^{2}
\leq
P_{\mathrm T},
\label{eq:PF_TLM_C1}\\
&
R_k^{\mathrm{TLM}}(\mathbf{W})
\geq
r,\quad
\forall k\in\mathcal K .
\label{eq:PF_TLM_C2}
\end{align}
\end{subequations}

\subsubsection{HDM} For the HDM architecture, the effective beamforming matrix can be defined as $ \mathbf{W}
    \triangleq
    \mathbf{F}\mathbf{D}$. Accordingly, the EE maximization problem for the HDM architecture is
\begin{subequations}
\label{eq:PF_HDM}
\begin{align}
(\mathcal{P}_3):\quad
\underset{\mathbf{F},\mathbf{D}}{\max}\quad
&
\eta^{\mathrm{HDM}}(\mathbf{F},\mathbf{D})\\
\text{s.t.}\quad
&
\varpi
\left\|
\mathbf{F}\mathbf{D}
\right\|_{\mathrm F}^{2}
\leq
P_{\mathrm T},
\label{eq:PF_HDM_C1}\\
&
\mathbf{F}^{\mathrm H}\mathbf{F}
\preccurlyeq
\mathbf{I}_{M},
\label{eq:PF_HDM_C2}\\
&
R_k^{\mathrm{HDM}}(\mathbf{F},\mathbf{D})
\geq
r,\quad
\forall k\in\mathcal K .
\label{eq:PF_HDM_C3}
\end{align}
\end{subequations}

The optimization problems $(\mathcal{P}_1)$, $(\mathcal{P}_2)$, and $(\mathcal{P}_3)$ are challenging to solve directly due to the nonconvexity of the objective functions and the rate constraints, the fractional form of the EE objective functions, and the coupling among the optimization variables in $(\mathcal{P}_3)$. In the following section, we develop SCA-based algorithms to tackle these problems. 

\vspace{-0.3cm}

\section{Proposed Solution} \label{sec:sol}

This section reformulates $(\mathcal{P}_1)$ in reduced dimension, develops its SCA-based solution, and extends the approach to $(\mathcal{P}_2)$ and $(\mathcal{P}_3)$. It also discusses convergence and computational complexity.

\vspace{-0.2cm}
\subsection{Reduced-Dimension Reformulation of Problem $(\mathcal{P}_1)$ } \label{sec:DR}

To reduce the computational complexity of $(\mathcal{P}_1)$, we exploit the fact that $\mathbf{W}$ interacts with the channels only through $\mathbf{h}_k^{\mathrm{H}}\mathbf{w}_i$. Thus, without loss of optimality, $\mathbf{W}$ lies in $\mathrm{span}(\mathbf{H})$, where $\mathbf{H}=[\mathbf{h}_1,\ldots,\mathbf{h}_K]$ and $V=\mathrm{rank}(\mathbf{H})\leq K$. Let $\mathbf{U}\in\mathbb{C}^{N\times V}$ be an orthonormal basis of $\mathrm{span}(\mathbf{H})$. We can then write $\mathbf{W}=\mathbf{U}\mathbf{Q},$ where $\mathbf{Q}\in\mathbb{C}^{V\times K}$ is the reduced-dimension optimization variable.
\footnote{The basis $\mathbf{U}$ can be obtained from the reduced SVD of $\mathbf{H}$ or its QR decomposition.}

Defining the effective channels as $\tilde{\mathbf{h}}_k=\mathbf{U}^{\mathrm{H}}\mathbf{h}_k\in\mathbb{C}^{V\times1}$, the SINR, user rate, sum rate, and EE can be equivalently expressed as functions of the reduced-dimension variable $\mathbf{Q}$ by substituting $\mathbf{W}=\mathbf{U}\mathbf{Q}$ into the corresponding expressions. Also, the received quantization noise power becomes $ \varsigma_k(\mathbf{Q}) = \tilde{\mathbf{h}}_k^{\mathrm{H}} \mathbf{U} \mathbf{Q} \mathbf{Q}^{\mathrm{H}} \mathbf{U}^{\mathrm{H}}  \tilde{\mathbf{h}}_k =\tilde{\mathbf{h}}_k^{\mathrm{H}} \mathbf{Q}\mathbf{Q}^{\mathrm{H}} \tilde{\mathbf{h}}_k$. Moreover,~\eqref{eq:PF_SLM_C2} can be equivalently written as
\begin{equation}
    \mathbf{W}^{\mathrm{H}}\mathbf{W} = \mathbf{Q}^{\mathrm{H}}\mathbf{U}^{\mathrm{H}}\mathbf{U} \mathbf{Q} = \mathbf{Q}^{\mathrm{H}} \mathbf{Q}
    \preccurlyeq
    \text{diag}(p_1,\dots,p_K).
\end{equation}

Thus, the reduced-dimension reformulation of $(\mathcal{P}_1)$ is
\begin{subequations}
\label{eq:RP_SLM}
\begin{align}
(\mathcal{P}_4): \  \underset{\mathbf{Q},\mathbf{P}}{\max}\  & \eta^{\mathrm{SLM}} (\mathbf{Q},\mathbf{P}) \label{eq:RSLM_obj} \\
\text{s.t.} \   & \varpi \sum_{k\in \mathcal{K}} p_k \leq P_{\mathrm{T}}, \label{eq:RSLM_PB} \\
& \mathbf{Q}^{\mathrm{H}}\mathbf{Q}
    \preccurlyeq
    \text{diag}(p_1,\dots,p_K),  \label{eq:RSLM_PSD}\\
& R_k^\mathrm{SLM} (\mathbf{Q}) \geq r \ \forall k \in \mathcal{K} .\label{eq:RSLM_QoS} 
\end{align}
\end{subequations}
This reduces the number of real-valued variables from $2NK+K$ to $2K^2+K$, yielding a substantial reduction when $K\ll N$. For example, with $N=64$ and $K=8$, the optimization dimension decreases from $1032$ to $136$, an $86.8\%$ reduction.

\vspace{-0.4cm}
\subsection{Proposed Solution for $(\mathcal{P}_4)$}
Next, we develop an SCA-based algorithm to find a stationary point of $(\mathcal{P}_4)$. At the $t$-th iteration, we first apply the Dinkelbach's algorithm to the fractional objective function~\eqref{eq:RSLM_obj}, yielding
\begin{equation}
\begin{split}
    \bar{\eta}^\mathrm{SLM}  (\mathbf{Q},\mathbf{P} ; & \mathbf{Q}^{(t-1)},  \mathbf{P}^{(t-1)}) = R_{\mathrm{sum}}^{\mathrm{SLM}} (\mathbf{Q}) \\
    & - \eta^\mathrm{SLM} (\mathbf{Q}^{(t-1)},\mathbf{P}^{(t-1)}) P_{\mathrm{tot}}^{\mathrm{SLM}} (\mathbf{P}) ,
    \end{split}
    \label{eq:DB}
\end{equation}
where $\mathbf{Q}^{(t)}$ and $\mathbf{P}^{(t)}$ are the values of $\mathbf{Q}$ and $\mathbf{P}$ at the $t$-th iteration of the proposed algorithm. We then express $R_{\mathrm{sum}}^{\mathrm{SLM}} (\mathbf{Q})$ as
\begin{equation}
\begin{split}
    & R_{\mathrm{sum}}^{\mathrm{SLM}} (\mathbf{Q})  = \sum_{k\in \mathcal{K}}  \log_2 \bigg(\varpi^2
    \sum_{i\in \mathcal{K}}
    |\tilde{\mathbf{h}}_k^{\mathrm{H}}\mathbf{q}_i|^2
    +
    \varsigma_k(\mathbf{Q})
     \\
    &+
    \sigma_k^2 \bigg) - \log_2 \bigg(
    \varpi^2
    \sum_{i \in \mathcal{K} \setminus \{ k \}}
    |\tilde{\mathbf{h}}_k^{\mathrm{H}}\mathbf{q}_i|^2
    +
    \varsigma_k(\mathbf{Q})
    +
    \sigma_k^2 \bigg).
    \end{split}
    \label{eq:RSLM_SR1}
\end{equation}
Next, we define the slack variables $\boldsymbol{\alpha} \triangleq [\alpha_{1,1}, \dots,\alpha_{K,K} ]^\mathrm{T}$, $\boldsymbol{\beta} \triangleq [\beta_{1}, \dots,\beta_{K} ]^\mathrm{T}$ and $\boldsymbol{\zeta} \triangleq [\zeta_{1}, \dots,\zeta_{K} ]^\mathrm{T}$ such that
\begin{subequations}
\begin{align}
    \alpha_{k,i} & \leq    |\tilde{\mathbf{h}}_k^{\mathrm{H}}\mathbf{q}_i|^2 \ \forall k,i \in \mathcal{K},
    \label{eq:NC1} \\
    \beta_{k} & \leq    \varsigma_k (\mathbf{Q}) \  \forall k \in \mathcal{K} ,
        \label{eq:NC2} \\
        \zeta_k & \geq   \varpi^2
    \sum_{i \in \mathcal{K} \setminus \{k \}}
    |\tilde{\mathbf{h}}_k^{\mathrm{H}}\mathbf{q}_i|^2
    +
    \varsigma_k(\mathbf{Q}) \ \forall k,i \in \mathcal{K}.
            \label{eq:NC3}
    \end{align}
\end{subequations}
Thus, we can lower-bound the sum rate as 
\begin{equation}
\begin{split}
    R_{\mathrm{sum}}^{\mathrm{SLM}} ( \mathbf{Q} )  \geq \sum_{k \in \mathcal{K}} \log_2(\varpi^2
    \sum_{i \in \mathcal{K}} &
    \alpha_{k,i}
    +
    \beta_k
    +
    \sigma_k^2) \\
    & - \log_2 (
    \zeta_k
    +
    \sigma_k^2).
    \end{split}
    \label{eq:RSLM_SR2}
\end{equation}
It can be noted that the first term in the left-hand side in~\eqref{eq:RSLM_SR2} is already concave, while the second term needs to be linearized. To do so, we use the following surrogate function derived from the first-order Taylor expansion:
\begin{equation}
\begin{split}
\tilde{R}_{\mathrm{sum}}^{\mathrm{SLM}} (  \boldsymbol{\alpha} , \boldsymbol{\beta},\boldsymbol{\zeta} ; \boldsymbol{\zeta}^{(t-1)} ) & = \sum_{k \in \mathcal{K}}  \log_2\bigg(\varpi^2
    \sum_{i \in \mathcal{K}}
    \alpha_{k,i}
    + 
    \beta_k
     \\
    & 
    + \sigma_k^2\bigg) - \frac{1}{\ln 2} \ \frac{(\zeta_k - \zeta_k^{(t-1)})}{\zeta_k^{(t-1)} + \sigma^2_k}.
    \end{split}
    \label{eq:conv_expr}
\end{equation}
 The expression~\eqref{eq:conv_expr} is now concave in the optimization variables. Hence, a concave surrogate function on~\eqref{eq:DB} can be obtained as
\begin{equation}
\begin{split}
    &\tilde{\eta}^\mathrm{SLM}  (\mathbf{P}, \boldsymbol{\alpha} , \boldsymbol{\beta},\boldsymbol{\zeta} ; \mathbf{Q}^{(t-1)},  \mathbf{P}^{(t-1)}, \boldsymbol{\zeta}^{(t-1)}) = \tilde{R}_{\mathrm{sum}}^{\mathrm{SLM}} (  \\
    & \boldsymbol{\alpha}, \boldsymbol{\beta} , \boldsymbol{\zeta} ; \boldsymbol{\zeta}^{(t-1)} ) - \eta^\mathrm{SLM} (\mathbf{Q}^{(t-1)},\mathbf{P}^{(t-1)}) P_{\mathrm{tot}}^{\mathrm{SLM}} (\mathbf{P}).
    \end{split}
    \label{eq:DB_SF}
\end{equation}

Next, we note that the new constraint~\eqref{eq:NC3} is convex, while appropriate surrogate functions are needed to tackle the right-hand sides of~\eqref{eq:NC1} and~\eqref{eq:NC2}. Using the first-order Taylor series expansion, we can find surrogate functions to approximate~\eqref{eq:NC1} and~\eqref{eq:NC2} as
\begin{subequations}
\begin{align}
\alpha_{k,i} &\!\leq \!  2\Re\{\mathbf{q}_i^{\mathrm{H}}\tilde{\mathbf{h}}_k\tilde{\mathbf{h}}_k^{\mathrm{H}}\mathbf{q}_i^{(t-1)}\}
\! -\! \left|\tilde{\mathbf{h}}_k^{\mathrm{H}}\mathbf{q}_i^{(t-1)}\right|^2 \! \! \! ,\ \forall i,k\in\mathcal{K} \!, \label{eq:NC4}\\
\beta_k &\leq 2\Re\{\tilde{\mathbf{h}}_k^{\mathrm{H}}\mathbf{Q}\mathbf{Q}^{(t-1)\mathrm{H}}\tilde{\mathbf{h}}_k\}
-\|\mathbf{Q}^{(t-1)\mathrm{H}}\tilde{\mathbf{h}}_k\|_2^2,\ \forall k\in\mathcal{K}, \label{eq:NC5}
\end{align}
\end{subequations}
which are both convex in the optimization variables.

Next, we tackle the constraints~\eqref{eq:RSLM_PB} and~\eqref{eq:RSLM_PSD}. While the former is already convex, the latter can be expressed in a convex linear matrix inequality (LMI) format as
\begin{equation}
\begin{bmatrix}
\text{diag} (p_1,\dots,p_K) & \mathbf{Q}^{\mathrm H}\\
\mathbf{Q} & \mathbf{I}_V
\end{bmatrix}
\succeq \mathbf{0},
\label{eq: SPD_conv}
\end{equation}
which is convex in $\mathbf{Q}$ and $\mathbf{P}$.

Next, we tackle the constraint set~\eqref{eq:RSLM_QoS} by first expressing the $k$-th user's rate constraint as
\begin{equation}
     |\tilde{\mathbf{h}}_k^{\mathrm{H}}\mathbf{q}_k|^2 \! \geq \! (2^r-1)\bigg(\varpi^2
    \! \! \! \! \sum_{i\in \mathcal{K} \setminus \{k\}}
    |\tilde{\mathbf{h}}_k^{\mathrm{H}}\mathbf{q}_i|^2
    +
    \varsigma_k(\mathbf{Q})
    +
    \sigma_k^2\bigg).
    \label{eq:QoS_for}
\end{equation}
We can observe that the right-hand side of~\eqref{eq:QoS_for} is convex, and the left-hand side can be tackled using~\eqref{eq:NC1}, leading to
\begin{equation}
    \alpha_{k,k} \geq (2^r-1)\bigg(\varpi^2
    \sum_{i\in \mathcal{K} \setminus \{k\}}
    |\tilde{\mathbf{h}}_k^{\mathrm{H}}\mathbf{q}_i|^2
    +
    \varsigma_k(\mathbf{Q})
    +
    \sigma_k^2\bigg) \ \forall k \in \mathcal{K}.
    \label{eq:NC6}
\end{equation}

Thus, the convex SCA subproblem at the $t$-th iteration is given by
\begin{subequations}
\label{eq:SCA_SLM}
\begin{align}  \underset{\mathbf{Q},\mathbf{P}, \boldsymbol{\alpha}, \boldsymbol{\beta}, \boldsymbol{\zeta}}{\max}\  & \tilde{\eta}^\mathrm{SLM}  (\mathbf{P},  \boldsymbol{\alpha} , \boldsymbol{\beta}, \boldsymbol{\zeta} ; \mathbf{Q}^{(t-1)},  \mathbf{P}^{(t-1)}, \boldsymbol{\zeta}^{(t-1)}) \label{eq:SOF} \\
\text{s.t.} \   & \eqref{eq:RSLM_PB}, \eqref{eq:NC3}, \eqref{eq:NC4}, \eqref{eq:NC5}, \eqref{eq: SPD_conv}, \eqref{eq:NC6}. 
\end{align}
\end{subequations}
This problem is convex and can be solved using the numerical optimization algorithms, e.g., via the CVX toolbox~\cite{2014_Grant}.
\vspace{-0.3cm}
\textbf{Algorithm~\ref{alg:SLM}} summarizes the proposed algorithm.

\begin{algorithm}[t]
\small
\caption{EE maximization in the SLM system.}
\label{alg:SLM}

\KwIn{$\mathbf{U}$, $\mathbf{Q}^{(0)}$, $\mathbf{P}^{(0)}$, 
$\zeta_k^{(0)} = \varpi^2
\sum_{i \in \mathcal{K} \setminus k}
|\tilde{\mathbf{h}}_k^{\mathrm{H}}\mathbf{q}^{(0)}_i|^2
+
\varsigma_k(\mathbf{Q}^{(0)})$ $\forall k,i \in \mathcal{K}$,
$t=0$, $\epsilon \geq 0$}

\Repeat{$
\eta_\mathrm{SLM}(\mathbf{Q}^{(t)},\mathbf{P}^{(t)})
-
\eta_\mathrm{SLM}(\mathbf{Q}^{(t-1)},\mathbf{P}^{(t-1)})
\leq \epsilon$}{
    Update $t \leftarrow t+1$\;

    Update $\mathbf{Q}^{(t)}$ and $\mathbf{P}^{(t)}$ by solving~\eqref{eq:SCA_SLM}\;
}

\KwOut{$\mathbf{W} = \mathbf{U}\mathbf{Q}^{(t)}$, $\mathbf{P} = \mathbf{P}^{(t)}$}

\end{algorithm}

\subsection{Proposed Solutions for $(\mathcal{P}_2)$ and $(\mathcal{P}_3)$} \label{sec:pro_TH}
\begin{table}[t]
\vspace{-0.35cm}
\centering
\caption{The proposed solution procedures for $(\mathcal{P}_1)$--$(\mathcal{P}_3)$. \vspace{-0.35cm}}
\label{tab:alg_comparison}
\setlength{\tabcolsep}{4pt}
\renewcommand{\arraystretch}{1.1}
\begin{tabular}{|l|c|c|c|}
\hline
Optimization problem & $(\mathcal{P}_1)$ & $(\mathcal{P}_2)$ & $(\mathcal{P}_3)$ \\
\hline
Optimization variables
& $\mathbf{Q},\mathbf{P}$
& ${\mathbf{Q}}$
& $\mathbf{Z},\mathbf{D}$ \\
\hline
Optimization strategy
& Joint SCA
& SCA
& Alternating SCA \\
\hline

LMI constraint
& Yes
& No
& Yes \\
\hline

Dimensionality retrieval
& $\mathbf{W}=\mathbf{U}\mathbf{Q}$
& $\mathbf{W}=\mathbf{U}{\mathbf{Q}}$
& $\mathbf{F}=\mathbf{U}\mathbf{Z}$ \\
\hline

Post-processing
& None 
& Via SVD
& None \\
\hline
\end{tabular}
\end{table}

The optimization problem $(\mathcal{P}_2)$ is a simplified version of $(\mathcal{P}_1)$, involving only a single optimization variable (i.e., $\mathbf{W}$) and no LMI constraint. Therefore, its solution can be obtained using a similar SCA framework, after applying the relevant modifications, to that described in Algorithm~\ref{alg:SLM}. The corresponding $\mathbf{F}$, $\boldsymbol{\Xi}$ and $\mathbf{P}$ are then obtained from the SVD of $\mathbf{W}$ as described in~\cite{2026_Zhou}.

On the other hand, due to the coupling between $\mathbf{F}$ and $\mathbf{D}$ in $(\mathcal{P}_3)$, an alternating optimization (AO) procedure is required. We first apply a dimensionality reduction procedure similar to that described in Section~\ref{sec:DR} by expressing $\mathbf{F} = \mathbf{U}\mathbf{Z}$. The algorithm then alternates between optimizing $\mathbf{D}$ and $\mathbf{Z}$ until convergence. Table~\ref{tab:alg_comparison} compares the proposed solution procedures for the three optimization problems.

\vspace{-0.4cm}

\subsection{Convergence Analysis}
Let $\mathcal{F}^{\mathrm{SLM}}
\triangleq
\{
(\mathbf{Q},\mathbf{P})
\mid
\eqref{eq:RSLM_PB}\text{--}\eqref{eq:RSLM_QoS}
\text{ are satisfied} \}$, denote the feasible set of $(\mathcal{P}_4)$. The transmit power, LMI, minimum-rate, and non-negativity constraints imply that $\mathcal{F}^{\mathrm{SLM}}$ is closed and bounded. Since $P_{\mathrm{tot}}^{\mathrm{SLM}}(\mathbf{P})\geq P_{\mathrm{cir}}^{\mathrm{SLM}}>0$, the EE is also upper bounded.

Algorithm~\ref{alg:SLM} employs the surrogate rate function $\tilde{R}_{\mathrm{sum}}^{\mathrm{SLM}}$. After restoring the omitted constant term, the resulting surrogate $\bar{R}_{\mathrm{sum}}^{\mathrm{SLM}}$ has the same maximizer as \eqref{eq:SCA_SLM} and satisfies the standard SCA conditions: it is tight at the current iterate, globally lower bounds the original sum rate, and has the same first-order derivative at the expansion point. Hence, the surrogate provides a valid approximation of the Dinkelbach-parametrized objective. Since the previous iterate remains feasible for the SCA subproblem, the resulting EE sequence satisfies $\eta^{\mathrm{SLM}}
\left(
\mathbf{Q}^{(t)},\mathbf{P}^{(t)}
\right)
\geq
\eta^{\mathrm{SLM}}
\left(
\mathbf{Q}^{(t-1)},\mathbf{P}^{(t-1)}
\right).$ Thus, the EE sequence is monotonically non-decreasing and, being upper bounded, converges to a finite limit. Under the standard regularity conditions for SCA and Dinkelbach-based fractional optimization, every accumulation point satisfies the Karush–Kuhn–Tucker (KKT) conditions of $(\mathcal{P}_4)$, owing to the tightness, lower-bound, and first-order consistency of the surrogate together with Dinkelbach's theorem, and hence also of the original problem $(\mathcal{P}_1)$.

The same argument applies to the TLM problem $(\mathcal{P}_2)$, which is solved using the same SCA-Dinkelbach framework. For HDM, the proposed AO algorithm alternately optimizes the MiLAC beamformer $\mathbf{Z}$ and digital beamformer $\mathbf{D}$ using the same SCA procedure. Each block update monotonically improves or preserves the EE. Since the EE is upper bounded, the AO sequence converges to a finite limit. Furthermore, as each block update converges to a stationary point of its corresponding subproblem, the overall algorithm converges to a block-wise stationary point of $(\mathcal{P}_3)$.

\vspace{-0.3cm}
\subsection{Computational Complexity Analysis}
\label{sec:complexity}

This subsection analyzes the computational complexity of the proposed EE optimization algorithms for the SLM, TLM, and HDM architectures. In all cases, the dominant computational cost comes from solving the convex subproblems generated by the SCA framework using a primal-dual interior-point method.

\subsubsection{SLM}

The proposed SLM algorithm first constructs the orthonormal basis $\mathbf{U}$ from the channel matrix using SVD or QR decomposition, which requires $\mathcal{O}(NK^2)$ operations. At each SCA iteration, the reduced-dimensional problem~\eqref{eq:SCA_SLM} is formulated as an SDP with $3K^2+3K$ real variables and an LMI of size $2K\times2K$. Hence, the overall complexity is $\mathcal{O}\bigg( 
NK^2+I_{\mathrm{SLM}}\sqrt{2K}
\Big((3K^2+3K)^3  +4K^2(3K^2+3K)^2
\Big.
\Big.
+(3K^2+3K)8K^3\Big)
\bigg),$ where $I_{\mathrm{SLM}}$ denotes the number of SCA iterations.

\subsubsection{TLM}

For the TLM architecture, the reduced-dimensional optimization problem involves only the beamforming matrix $\mathbf{W}$ and does not contain semidefinite constraints. The reduced-dimensional beamforming variable has $2K^2$ real variables, resulting in a per-iteration complexity of $\mathcal{O}(\sqrt{K}(2K^2)^3)$. After convergence, recovering the two MiLAC beamformers and the power allocation matrix from $\mathbf{W}$ using SVD introduces an additional complexity of $\mathcal{O}(NK^2)$. Therefore, the overall complexity is $\mathcal{O}\left(
I_{\mathrm{TLM}}\sqrt{K}(2K^2)^3
+NK^2
\right)$, where $I_{\mathrm{TLM}}$ denotes the number of SCA iterations.

\subsubsection{HDM}

The proposed HDM algorithm employs an AO framework that alternately optimizes the MiLAC and digital beamformers. The reduced-dimensional MiLAC beamformer has $2KM$ real variables and its optimization is an SDP with an LMI of size $2K\times2K$, while the digital beamformer update has $2MK$ real variables. Thus, the overall complexity is $\mathcal{O}\bigg(
NK^2+I_{\mathrm{HDM}}\bigg[
\sqrt{2K}\Big((2KM)^3
+4K^2(2KM)^2
\Big.
\Big.
+(2KM)8K^3\Big)
+\sqrt{K}(2MK)^3
\bigg]\bigg),$ where $I_{\mathrm{HDM}}$ denotes the number of AO iterations.

The proposed algorithms have computational complexity that depends primarily on the reduced dimensions $K$ and $M$ rather than on the number of antennas $N$. Since practical large-scale MIMO systems typically satisfy $K\ll N$ and $M\ll N$, the proposed dimensionality reduction significantly improves scalability while preserving the optimality of the reduced-dimensional formulations.

\vspace{-0.3cm}
\section{Asymptotic Analysis of the Energy Efficiency} \label{sec:Asy_ana}

In this section, we characterize the EE scaling of the SLM, TLM, HDM, and DBF architectures in the large-antenna regime. We first develop a low-complexity search-based solution for SLM and characterize its EE scaling as the number of transmit antennas grows. We then extend the analysis to TLM, HDM, and DBF.

\vspace{-0.3cm}

\subsection{Low-Complexity Search-Based Solution for SLM }

We assume that the users' channels become asymptotically orthogonal, i.e., $\mathbf{h}_k^\mathrm{H}\mathbf{h}_i \approx 0$ for all $i,k\in\mathcal{K}$ with $i\neq k$, which is commonly adopted in massive MIMO analysis due to the favorable propagation property~\cite{2014_Rusek}. This assumption is particularly relevant for MiLAC-based beamforming, whose objective is to reduce the hardware complexity and power consumption of large-scale antenna arrays. Moreover, we assume high-resolution DACs at the BS, such that quantization noise is negligible. Accordingly, we set $\varpi=1$, which is achieved with high accuracy for $b\geq 3$~\cite{2019_Dai}.

Under the aforementioned assumptions, the inter-user interference is completely eliminated, and the SLM beamforming design decouples into independent single-user beamforming problems. Therefore, the optimal beamforming direction for each user is the one that maximizes the received signal power, i.e., $\max_{\|\mathbf{w}_k\|_2^2=\xi_k} |\mathbf{h}_k^\mathrm{H}\mathbf{w}_k|^2 .$ The maximum is achieved when $\mathbf{w}_k$ is aligned with $\mathbf{h}_k$. Hence, the optimal SLM beamforming vector for the $k$-th user is given by maximum ratio transmission (MRT) as
\begin{equation}
\mathbf{w}_k^{\mathrm{SLM}} = \sqrt{\xi_k} \frac{\mathbf{h}_k}{\lVert \mathbf{h}_k \rVert_2},
\end{equation}
where $\xi_k \geq 0$ is the power allocation coefficient that controls the transmit power assigned to the $k$-th user.

As a result, the sum rate, total power consumption, and EE can be expressed as
\begin{equation}
    R^{\mathrm{SLM}}_{\mathrm{sum}} (\boldsymbol{\xi}) =
    \sum_{k \in \mathcal{K}} \log_2 \left(1+\xi_k\omega_k\right),
\end{equation}
\begin{equation}
    P^{\mathrm{SLM}}_{\mathrm{tot}} (\boldsymbol{\xi}) =
    \frac{1}{\rho}\sum_{k\in\mathcal{K}}\xi_k+P^{\mathrm{SLM}}_{\mathrm{cir}},
\end{equation}
\begin{equation}
    \eta^{\mathrm{SLM}}(\boldsymbol{\xi}) =
    \frac{R^{\mathrm{SLM}}_{\mathrm{sum}}(\boldsymbol{\xi})}
    {P^{\mathrm{SLM}}_{\mathrm{tot}}(\boldsymbol{\xi})},
\end{equation}
where $\boldsymbol{\xi}\triangleq[\xi_1,\dots,\xi_K]^{\mathrm{T}}$ and
$\omega_k\triangleq \|\mathbf{h}_k\|_2^2/\sigma^2$.

Under the orthogonal channel assumption, constraint~\eqref{eq:PF_SLM_C2} is automatically satisfied. Therefore, the problem reduces to optimizing the power allocation coefficients. To characterize the optimal power allocation, we distinguish between users whose minimum rate constraint is active and those receiving additional rate beyond the minimum requirement. So, we introduce a nonnegative rate margin $\bar{r}_k\geq0$ such that
\begin{equation}
    R_k^\mathrm{SLM}(\boldsymbol{\xi})=r+\bar r_k,
    \label{eq:new_rate}
\end{equation}
the required power allocation for user $k$ is given by
\begin{equation}
    \xi_k=\frac{2^{r+\bar r_k}-1}{\omega_k}.
\end{equation}
Consequently, $(\mathcal{P}_1)$ can be equivalently reformulated as
\begin{subequations}
\label{eq:PA_SLM_eq}
\begin{align}
  \underset{\bar{\mathbf r}}{\max}\quad&
  \frac{Kr+\sum_{k\in\mathcal K}\bar r_k}
  {\frac{1}{\rho}\sum_{k \in \mathcal{K}} \left(\frac{2^{r+\bar r_k}-1}{\omega_k}\right) +P_\mathrm{cir}^\mathrm{SLM}}\} \label{eq:obj_lemma1} \\
\text{s.t.}\quad&
\sum_{k\in\mathcal K} \frac{2^{r+\bar r_k}-1}{\omega_k} \leq P_\mathrm{T}, \label{eq:PBC_lemma1}  \\
&\bar{\mathbf r}\geq \mathbf{0}_{K \times 1},
\end{align}
\end{subequations}
where $\bar{\mathbf r}=[\bar r_1,\dots,\bar r_K]^\mathrm{T}$. We then introduce the following lemma.

\begin{lemma}\label{lemma:1}
Let $\mathcal{K}'\subseteq\mathcal{K}$ denote the nonzero rate-margin set containing the $K'\leq K$ users satisfying $\bar{r}_k>0$, where $|\mathcal{K}'|=K'$. The optimal EE is then given by
\begin{equation}
\eta^\mathrm{SLM} = \frac{K'}{A_{\mathcal{K}'}\ln(2)} W\left( \frac{A_{\mathcal{K}'}\ln(2)}{K'} \exp\left( \frac{B_{\mathcal{K}'}\ln(2)}{K'} - 1 \right) \right),
\label{eq:EE_lm1}
\end{equation}
where $W(\cdot)$ denotes the Lambert-$W$ function, and
\begin{subequations}
\begin{align}
A_{\mathcal{K}'} & \triangleq P_{\mathrm{cir}}^{\mathrm{SLM}} + \frac{1}{\rho} \sum_{k\notin\mathcal{K}'} \frac{2^r-1}{\omega_k} - \frac{1}{\rho} \sum_{k\in\mathcal{K}'} \frac{1}{\omega_k}, \\
B_{\mathcal{K}'} & \triangleq (K-K')r + \sum_{k\in\mathcal{K}'} \log_2(\omega_k) + K'\log_2\left(\frac{\rho}{\ln(2)}\right).
\end{align}
\end{subequations}
The corresponding optimal rate margins are
\begin{equation}
\bar r_k =
\begin{cases}
\log_2 \left( \dfrac{\rho\omega_k}{\lambda\ln(2)} \right)-r, & k\in\mathcal{K}',\\[1ex]
0, & k\notin\mathcal{K}'.
\end{cases}
\end{equation}

The obtained solution is feasible if and only if
\begin{equation}
\sum_{k\in\mathcal{K}}\frac{2^{\bar r_k}}{\omega_k}
\leq
2^{-r}\left(
P_{\mathrm{T}}
+
\sum_{k\in\mathcal{K}}\frac{1}{\omega_k}
\right).
\label{eq:FC_lm1}
\end{equation}
For $K'=0$, the expression~\eqref{eq:EE_lm1} is $\lim_{K'\rightarrow0}\eta^{\mathrm{SLM}}=B_{\emptyset}/A_{\emptyset}$,
which reduces to
\begin{equation}
\eta^{\mathrm{SLM}}\big|_{K'=0}
=
\frac{Kr}
{\frac{1}{\rho}\sum_{k\in\mathcal{K}}\frac{2^r-1}{\omega_k}+P_{\mathrm{cir}}^{\mathrm{SLM}}},
\label{eq:EE_K0}
\end{equation}
i.e., all users are served at exactly their minimum rate $r$.

\end{lemma}

\begin{proof}
By applying the Dinkelbach transformation to the objective function~\eqref{eq:obj_lemma1}, the fractional objective can be equivalently solved by maximizing the following Lagrangian function:
\begin{equation}
\mathcal{L} (\bar{\mathbf{r}},\lambda) =   Kr+\sum_{k\in\mathcal{K}'}\bar r_k - \lambda \left( \frac{1}{\rho} \sum_{k\in\mathcal{K}} \frac{2^{r+\bar r_k}-1}{\omega_k} + P_{\mathrm{cir}}^{\mathrm{SLM}} \right).
\end{equation}
For each user in the set $\mathcal{K}'$, the stationarity condition is
\begin{equation}
 1 - \lambda \frac{\ln(2)2^{r+\bar r_k}}{\rho\omega_k}=0 
\Rightarrow \bar r_k = \log_2 \left( \frac{\rho\omega_k}{\lambda\ln(2)} \right)-r.
\end{equation}
Substituting the above expression into the sum rate expression yields
\begin{equation}
\begin{aligned}
& R_\mathrm{sum}^{\mathrm{SLM}} (\lambda) = Kr+ \sum_{k\in\mathcal{K}'} \left( \log_2 \left( \frac{\rho\omega_k}{\lambda\ln(2)} \right)-r \right) \\
&= (K-K')r + \sum_{k\in\mathcal{K}'} \log_2(\omega_k) + K'\log_2 \left( \frac{\rho}{\lambda\ln(2)} \right) \\
&= B_{\mathcal{K}'} - K'\log_2(\lambda).
\end{aligned}
\end{equation}
Similarly, the total power consumption becomes
\begin{equation}
\begin{split}
P_\mathrm{tot}^{\mathrm{SLM}} (\lambda)  &= \frac{1}{\rho} \sum_{k\in\mathcal{K}'} \frac{\frac{\rho\omega_k}{\lambda\ln(2)}-1}{\omega_k} + \frac{1}{\rho} \sum_{k\notin\mathcal{K}'} \frac{2^r-1}{\omega_k} \\
& + P_{\mathrm{cir}}^{\mathrm{SLM}}  = \frac{K'}{\lambda\ln(2)} + A_{\mathcal{K}'}.
\end{split}
\end{equation}

Moreover, the optimal Dinkelbach parameter should satisfy $\lambda^{\mathrm{opt}} =R_\mathrm{sum}^{\mathrm{SLM}} (\lambda^{\mathrm{opt}}) /P_\mathrm{tot}^{\mathrm{SLM}} (\lambda^{\mathrm{opt}}) = \eta^\mathrm{SLM}$ . Hence,
\begin{equation}
\begin{split}
&\lambda^{\mathrm{opt}} \left( \frac{K'}{\lambda^{\mathrm{opt}} \ln(2)} + A_{\mathcal{K}'} \right) = B_{\mathcal{K}'} - K'\log_2(\lambda^{\mathrm{opt}} ) \\
& \Rightarrow
\lambda^{\mathrm{opt}}  = \frac{K'}{A_{\mathcal{K}'}\ln(2)} W\Bigg( \frac{A_{\mathcal{K}'}\ln(2)}{K'} \\
& \times \exp\left( \frac{B_{\mathcal{K}'}\ln(2)}{K'} - 1 \right) \Bigg) = \eta^{\mathrm{SLM}}.
\end{split}
\end{equation}
Finally, the feasibility of the solution follows directly by substituting the optimal rate margins into the transmit power constraint~\eqref{eq:PBC_lemma1}, which completes the proof.
\end{proof}

To determine the set $\mathcal{K}'$ in Lemma~\ref{lemma:1}, we employ a descending iterative search strategy. First, the users are sorted according to their effective channel gains in descending order, i.e., $\omega_1 \geq \omega_2 \geq \cdots \geq \omega_K$. The search is initialized with the largest possible rate-margin set, $\mathcal{K}'=\mathcal{K}$ (i.e., $K'=K$), where all users are assumed to have positive rate margins. For this candidate set, the corresponding EE is computed using Lemma~\ref{lemma:1}, and the associated transmit power feasibility condition is verified. If the candidate set satisfies the feasibility condition, it is retained for comparison. The rate-margin set is then updated by removing the user with the weakest effective channel gain, and the same procedure is repeated until the final case of $K'=0$ is reached. Among all feasible candidate sets, the set achieving the maximum EE is selected as the optimal solution. The complete procedure is summarized in \textbf{Algorithm~\ref{alg:rate_margin_SLM}}.

\begin{algorithm}[t]
\small
\caption{Low-dimensional search-based closed-form approximate solution for the SLM system.}
\label{alg:rate_margin_SLM}

\KwIn{$\{\omega_k\}_{k\in\mathcal{K}}$, $r$, $P_{\mathrm{T}}$, $P_{\mathrm{cir}}^{\mathrm{SLM}}$, $\rho$}

Sort the users such that
$\omega_1 \geq \omega_2 \geq \cdots \geq \omega_K$\;

Initialize $\eta^{\mathrm{opt}}=0$ and $\mathcal{K}^{\mathrm{opt}}=\emptyset$\;

\For{$K'=K$ \KwTo $0$}{

    Set $\mathcal{K}'=\{1,\dots,K'\}$ and compute $A_{\mathcal{K}'}$ and $B_{\mathcal{K}'}$\;

    \eIf{$K'=0$}{
        Compute $\eta^{\mathrm{SLM}}$ using~\eqref{eq:EE_K0}\;
    }{
        Compute $\eta^{\mathrm{SLM}}$ using~\eqref{eq:EE_lm1}\;
    }

    Compute the rate margins $\bar{r}_k$, $\forall k \in \mathcal{K}$\;

    \If{the feasibility condition~\eqref{eq:FC_lm1} is satisfied
        \textbf{and} $\eta^{\mathrm{SLM}}>\eta^{\mathrm{opt}}$}{

        Update $\eta^{\mathrm{opt}}\leftarrow\eta^{\mathrm{SLM}}$ and
        $\mathcal{K}^{\mathrm{opt}}\leftarrow\mathcal{K}'$\;
    }
}

\KwOut{$\eta^{\mathrm{SLM}}=\eta^{\mathrm{opt}}$, $\mathcal{K}'=\mathcal{K}^{\mathrm{opt}}$}

\end{algorithm}

\begin{figure}
  \centering
  \begin{tabular}{c c}
    \hspace{-0.3cm} \includegraphics[width=0.55\columnwidth]{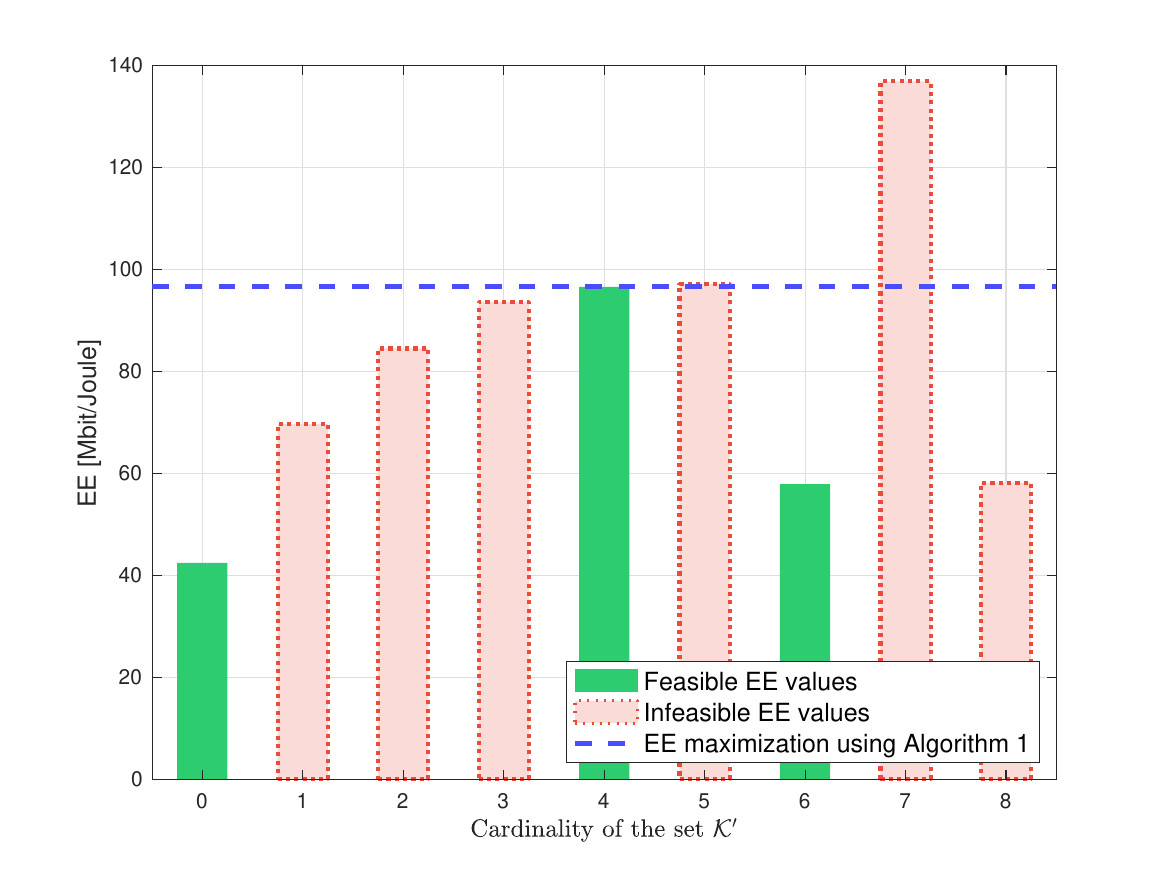} & \hspace{-0.8cm}\includegraphics[width=0.55\columnwidth]{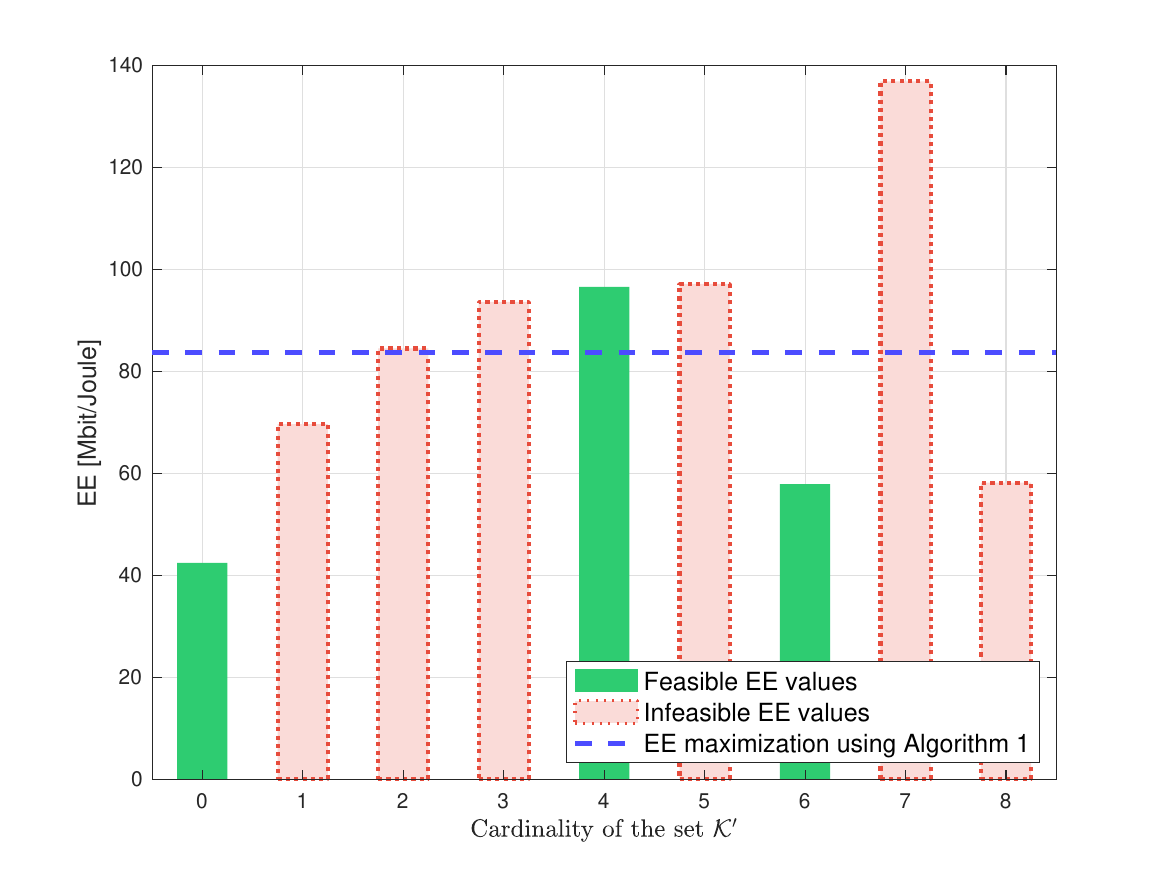} \\
    \hspace{-0.3cm}  \scriptsize (a)  & 
\hspace{-0.8cm} \scriptsize (b) 
  \end{tabular}
    \medskip
    \vspace{-0.6cm}
  \caption{ Outputs of Algorithm~\ref{alg:rate_margin_SLM} and Algorithm~\ref{alg:SLM} for (a) orthogonal channels without quantization noise and (b) non-orthogonal multipath channels with quantization noise.}
  \label{fig:SLM_EE_LN}
\end{figure}

To illustrate the proposed search procedure, Fig.~\ref{fig:SLM_EE_LN} shows the EE obtained for each candidate cardinality $K'$ of $\mathcal{K}'$. Thus, each bar represents a candidate solution with $K'$ users receiving a rate above their minimum requirement, while the remaining $K-K'$ users achieve exactly their minimum rate. Green and red bars indicate feasible and infeasible solutions, respectively, according to~\eqref{eq:FC_lm1}, while the dashed line indicates the EE obtained by Algorithm~\ref{alg:SLM}. For the considered example, the candidate with the highest feasible EE corresponds to $K'=4$. For orthogonal channels without quantization noise, the maximum feasible EE obtained from the candidate search coincides with the EE of Algorithm~\ref{alg:SLM}, confirming that the proposed search achieves the optimal solution under these assumptions. For non-orthogonal channels with quantization noise, the candidate search provides an optimistic EE estimate because it neglects multi-user interference and quantization noise. Nevertheless, it provides useful insight into the achievable EE and the impact of the number of users receiving a rate margin without solving the original non-convex problem.

\vspace{-0.3cm}

\subsection{Asymptotic EE Scaling Analysis for the SLM System} \label{sec:scaling}

To gain further insights into the asymptotic EE behavior of the SLM architecture, we consider the large-antenna regime where the number of transmit antennas \(N\) increases while the nonzero rate-margin set \(\mathcal{K}'\) remains fixed. For a large-scale antenna array channel, the channel hardening property yields
\begin{equation}
\frac{1}{N}\|\mathbf{h}_k\|_2^2
\xrightarrow[N\rightarrow\infty]{\mathrm{a.s.}}
\chi_k,
\end{equation}
where \(\chi_k\) denotes the asymptotic normalized channel power gain~\cite{2016_Marzetta,2013_Ngo}. Accordingly, the effective channel-to-noise ratio can be approximated as $\omega_k
\approx
N\bar{\chi}_k$, where \(\bar{\chi}_k\triangleq\chi_k/\sigma_k^2\). Substituting this approximation into the expressions of \(A_{\mathcal{K}'}\) and \(B_{\mathcal{K}'}\), together with the circuit power model in~\eqref{eq:SLM_p_cir}, yields
\begin{equation}
\begin{aligned}
A_{\mathcal{K}'}
=&\,
\frac{P_{\mathrm{IT}}}{2}N^2
+
\frac{(2K+1)P_{\mathrm{IT}}}{2}N
+
\Gamma_{1}\\
&+
\frac{1}{\rho}
\sum_{k\notin\mathcal{K}'}
\frac{2^r-1}{N\bar{\chi}_k}
-
\frac{1}{\rho}
\sum_{k\in\mathcal{K}'}
\frac{1}{N\bar{\chi}_k},
\end{aligned}
\end{equation}
\begin{equation}
B_{\mathcal{K}'}
=
K'\log_2(N)+\Gamma_2,
\end{equation}
where \(\Gamma_1\) and \(\Gamma_2\)  collect all terms that are independent of \(N\).

As \(N\rightarrow\infty\), the last two terms in \(A_{\mathcal{K}'}\) vanish, while the quadratic term dominates the remaining terms. Hence,
\begin{equation}
A_{\mathcal{K}'}
\approx
\frac{P_{\mathrm{IT}}}{2}N^2.
\end{equation}
Substituting this approximation into~\eqref{eq:EE_lm1} gives
\begin{equation}
\eta^{\mathrm{SLM}}
\approx
\frac{2K'}
{P_{\mathrm{IT}}N^2\ln(2)}
W\!\left(
\Gamma_3 N^3
\right),
\end{equation}
where $\Gamma_3 \triangleq
\frac{P_{\mathrm{IT}}\ln(2)}
{2K'}
\exp\!\left(
\frac{\Gamma_2\ln(2)}
{K'}
-1
\right)$ is independent of \(N\).

Using the asymptotic expansion of the Lambert-$W$ function
$W(x)\approx\ln(x)-\ln(\ln(x))$ for large $x$~\cite{1996_Corless}, we obtain
\begin{equation}
W(\Gamma_3 N^3)
\approx
3\ln(N)-\ln\!\bigl(3\ln(N)\bigr)+\ln(\Gamma_3).
\end{equation}
Therefore
\begin{equation}
\eta^{\mathrm{SLM}}
\approx
\frac{2K'}{P_{\mathrm{IT}}N^2\ln(2)}
\left[
3\ln(N)-\ln\!\bigl(3\ln(N)\bigr)+\ln(\Gamma_3)
\right].
\end{equation}
Hence, the optimal EE of the SLM architecture exhibits the asymptotic scaling
\begin{equation}
\eta^{\mathrm{SLM}}
\propto
\frac{\ln(N)}{N^2},
\end{equation}
as \(N\rightarrow\infty\), with an asymptotic scaling factor of $\frac{6K'}{P_{\mathrm{IT}}\ln(2)}$.

This result indicates that the EE of the SLM architecture eventually decreases at a rate of $\ln(N)/N^2$ as the antenna array size increases. This behavior is primarily attributed to the quadratic scaling of circuit power consumption with $N$, which arises from the increasing number of tunable impedance elements in the considered SLM architecture as demonstrated in~\eqref{eq:SLM_p_cir}.

\vspace{-0.4cm}

\subsection{Optimal Number of Transmit Antennas for SLM} \label{sec:peak}

While Section~\ref{sec:scaling} characterizes the asymptotic decay of the EE as $N\rightarrow\infty$, it does not determine whether the EE has a unique maximum or characterize its value. This is addressed in the following lemma.

\begin{lemma}\label{lemma:2}
Let the nonzero rate-margin set $\mathcal{K}'$ be approximately fixed. The EE admits a unique maximum at $N^{\mathrm{opt}}$ that is obtained by solving the equation
\begin{equation}
W\!\left(\kappa N^{\mathrm{opt}} A_{\mathcal{K}'}(N^{\mathrm{opt}})\right)
=
\frac{A_{\mathcal{K}'}(N^{\mathrm{opt}})}{N^{\mathrm{opt}}\,P_{\mathrm{IT}}\!\left(N^{\mathrm{opt}}+\frac{2K+1}{2}\right)},
\label{eq:Nopt_implicit_full}
\end{equation}
where $\kappa\triangleq\frac{\ln(2)}{K'}\exp\!\left(\frac{\Gamma_2\ln(2)}{K'}-1\right)$. The corresponding peak EE value is
\begin{equation}
\eta^{\mathrm{SLM}}(N^{\mathrm{opt}})\approx \frac{K'}{\ln(2)\,N^{\mathrm{opt}}P_{\mathrm{IT}}\!\left(N^{\mathrm{opt}}+\frac{2K+1}{2}\right)}.
\label{eq:eta_at_peak_full}
\end{equation}
\end{lemma}

\begin{proof}
From Lemma~\ref{lemma:1}, the optimal EE satisfies $\eta^{\mathrm{SLM}}P^{\mathrm{SLM}}_{\mathrm{tot}}=R^{\mathrm{SLM}}_{\mathrm{sum}}$, which upon implicit differentiation in $N$ requires $\partial B_{\mathcal{K}'}/\partial N=\eta^{\mathrm{SLM}}\partial A_{\mathcal{K}'}/\partial N$. Differentiating $A_{\mathcal{K}'}(N)$ gives $\partial A_{\mathcal{K}'}/\partial N=P_{\mathrm{IT}}\!\left(N+\frac{2K+1}{2}\right)$, and with $\partial B_{\mathcal{K}'}/\partial N=K'/(N\ln(2))$, this yields exactly~\eqref{eq:eta_at_peak_full}. Equating this with the exact Lemma~\ref{lemma:1} expression $\eta^{\mathrm{SLM}}=\frac{K'}{A_{\mathcal{K}'}(N)\ln(2)}W(\kappa N A_{\mathcal{K}'}(N))$ yields~\eqref{eq:Nopt_implicit_full}. $W(\cdot)$ is strictly increasing while the right-hand side of~\eqref{eq:Nopt_implicit_full} is strictly decreasing in $N$, so the crossing point is unique.
\end{proof}

\begin{remark}
Equation~\eqref{eq:Nopt_implicit_full} can be solved straightforwardly using a simple one-dimensional search. Since the crossing point is unique, a linear search over $N$ is sufficient to determine $N^{\mathrm{opt}}$.
\end{remark}

\subsection{DBF, TLM and HDM}

A similar asymptotic analysis can be carried out for the TLM, HDM, and DBF architectures by replacing the circuit power model \(P_{\mathrm{cir}}^{\mathrm{SLM}}\) in Lemma~\ref{lemma:1} with \(P_{\mathrm{cir}}^{\mathrm{TLM}}\), \(P_{\mathrm{cir}}^{\mathrm{HDM}}\), and \(P_{\mathrm{cir}}^{\mathrm{DBF}}\), respectively, and subsequently applying Algorithm~\ref{alg:rate_margin_SLM}. This follows from the fact that, under the asymptotic orthogonality of user channels, none of the considered architectures imposes an inherent limitation on the achievable beamforming gain. Furthermore, the TLM architecture can realize any DBF beamformer, while the HDM architecture further generalizes the TLM by replacing the MiLAC front-end with a digital beamforming stage. Consequently, all four architectures achieve the same asymptotic beamforming gain. Nevertheless, their optimal EE differs due to the different circuit power consumption associated with each architecture\footnote{While Lambert-$W$-based EE solutions have been studied for conventional RF-chain-based systems~\cite{2012_Isheden,2014_Huang}, MiLAC architectures introduce different circuit power models and EE optimization structures. Our work further develops a unified framework for SLM, TLM, and HDM, derives low-dimensional closed-form search solutions, and establishes their asymptotic EE scaling to characterize the impact of MiLAC architectures and power consumption on large-scale antenna systems.}.

Following similar derivations as in Section~\ref{sec:scaling}, the asymptotic EE scaling laws can be obtained for the DBF, TLM and HDM architectures as follows:
\begin{equation}
\eta^{\mathrm{DBF}} \propto \frac{\ln(N)}{N},\qquad
\eta^{\mathrm{TLM}},\eta^{\mathrm{HDM}} \propto \frac{\ln(N)}{N^2},
\end{equation}
as \(N\rightarrow\infty\), with scaling factors \(\frac{2K'}{P_{\mathrm{RF}}\ln(2)}\) for the DBF and \(\frac{6K'}{P_{\mathrm{IT}}\ln(2)}\) for both the TLM and HDM, where \(P_{\mathrm{RF}} \triangleq P_\mathrm{H}+2(P_\mathrm{VGA}+P_\mathrm{M}+P_\mathrm{LP}+ P_\mathrm{ADC})\).

These scaling laws indicate that, under the considered power consumption models, the EE of the MiLAC-based architectures (SLM, TLM, and HDM) decreases more rapidly than that of DBF in the asymptotic large-antenna regime. This difference arises from the quadratic growth of the circuit power in the considered MiLAC architectures, which results from the increasing number of tunable impedance elements in the analog network, whereas the DBF circuit power grows linearly with $N$. Importantly, this asymptotic behavior is specific to the adopted architecture and power consumption models and does not imply that MiLAC-based architectures are less energy efficient than DBF for practically relevant antenna sizes. In fact, as shown by our numerical results, MiLAC-based architectures can provide substantially higher EE over practical antenna regimes. The asymptotic scaling should therefore be interpreted as a characterization of the eventual trend for very large $N$, rather than as a direct comparison of the absolute EE of the two architectures.

The characterization in Section~\ref{sec:peak} extends to the other architectures by replacing $A_{\mathcal{K}'}(N)$ with the corresponding circuit power model. For TLM and HDM, the optimal antenna dimension $N^{\mathrm{opt}}$ is obtained by solving
\begin{equation}
W\!\left(\kappa N^{\mathrm{opt}} A_{\mathcal{K}'}(N^{\mathrm{opt}})\right)
=
\frac{A_{\mathcal{K}'}(N^{\mathrm{opt}})}
{N^{\mathrm{opt}}P_{\mathrm{IT}}
\left(N^{\mathrm{opt}}+\frac{2K'+1}{2}\right)}.
\label{eq:Nopt_TLM_HDM}
\end{equation}
The corresponding peak EE is
\begin{equation}
\eta^{\mathrm{opt}}
\approx
\frac{K'}{\ln(2)\,N^{\mathrm{opt}}P_{\mathrm{IT}}
\left(N^{\mathrm{opt}}+\frac{2K'+1}{2}\right)}.
\label{eq:eta_TLM_HDM}
\end{equation}
For DBF, the optimal antenna dimension is instead obtained from
\begin{equation}
W\!\left(\kappa N^{\mathrm{opt}} A_{\mathcal{K}'}(N^{\mathrm{opt}})\right)
=
\frac{A_{\mathcal{K}'}(N^{\mathrm{opt}})}
{N^{\mathrm{opt}}P_{\mathrm{RF}}},
\label{eq:Nopt_DBF}
\end{equation}
with the corresponding peak EE given by
\begin{equation}
\eta^{\mathrm{opt}}
\approx
\frac{K'}{\ln(2)\,N^{\mathrm{opt}}P_{\mathrm{RF}}}.
\label{eq:eta_DBF}
\end{equation}
\vspace{-0.5cm}

\section{Numerical Simulations} \label{sec:sim}
In this section, we provide numerical results to evaluate the performance of the considered architectures and validate the effectiveness of the proposed MiLAC-based designs. Specifically, we compare the three MiLAC architectures, namely SLM, TLM, and HDM, with conventional DBF and HAD systems in terms of their achievable EE.

We use the following parameters unless otherwise stated. The number of transmit antennas and receive users are set to \(N=64\) and \(K=8\), respectively. For the HDM and HAD systems, the number of RF chains is set to \(M=12\). The system carrier frequency, bandwidth, and DAC sampling frequency at the BS are set to \(f_\mathrm{c}=\unit[28]{GHz}\), \(BW=\unit[20]{MHz}\), and \(f_\mathrm{s}=\unit[1]{GHz}\), respectively. The RF chains are assumed to employ DACs with \(b=4\) quantization bits. The static power consumption of each hardware component is set according to Table~\ref{tab:2}, and the power amplifier efficiency is set to \(\rho=0.27\). The transmit power budget and noise power are set to \(P_\mathrm{T}=\unit[35]{dBm}\) and \(\sigma_k^2=\unit[-93]{dBm}\), respectively, for all \(k\in\mathcal{K}\). Finally, the minimum rate requirement is set to \(r=\unit[1]{bit/sec/Hz}\).

The channels are modeled using a Saleh--Valenzuela (SV) multipath mmWave channel model. Specifically, a geometric channel model with \(L=5\) multipath components is considered, where the channel consists of a sum of randomly generated paths with independent complex Gaussian path gains and uniformly distributed angles of departure. The large-scale fading is modeled using a distance-dependent path-loss model with a path-loss exponent of \(3.3\). The user distances are uniformly distributed between \(20\) and \(40\)~m. Finally, Monte Carlo simulations are conducted, and all simulation results are averaged over 100 independent random channel realizations.

\vspace{-0.3cm}
\subsection{Convergence Analysis}
\begin{figure*}[t]
    \centering
    
    \begin{minipage}{.32\textwidth}
        \centering
        \includegraphics[width=\columnwidth]{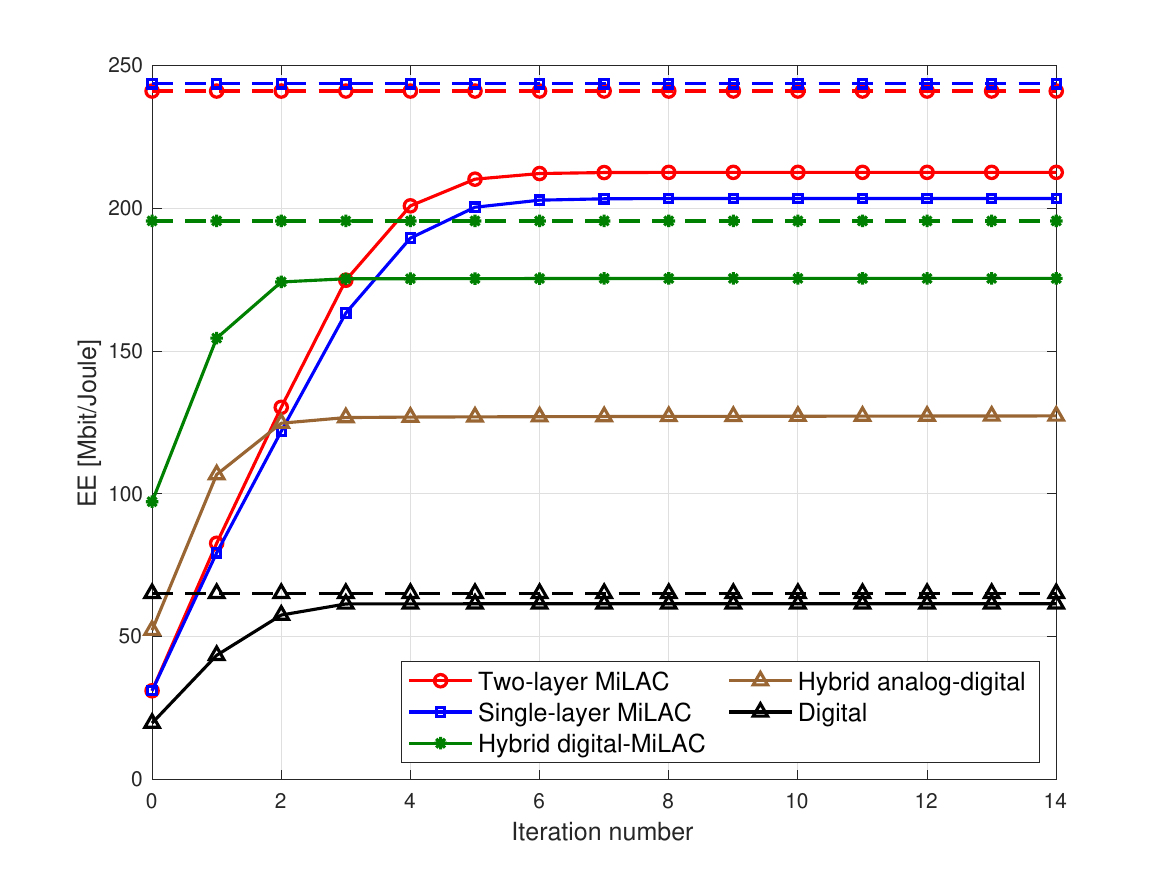}
        \vspace{-0.8cm}
        \caption{EE versus the number of iterations for the proposed algorithm (solid) and the low-dimensional closed-form search (dashed) in Section~\ref{sec:Asy_ana}.}
        \label{fig:conv}
    \end{minipage}%
    \hfill
    \begin{minipage}{.32\textwidth}
        \centering
        \includegraphics[width=\columnwidth]{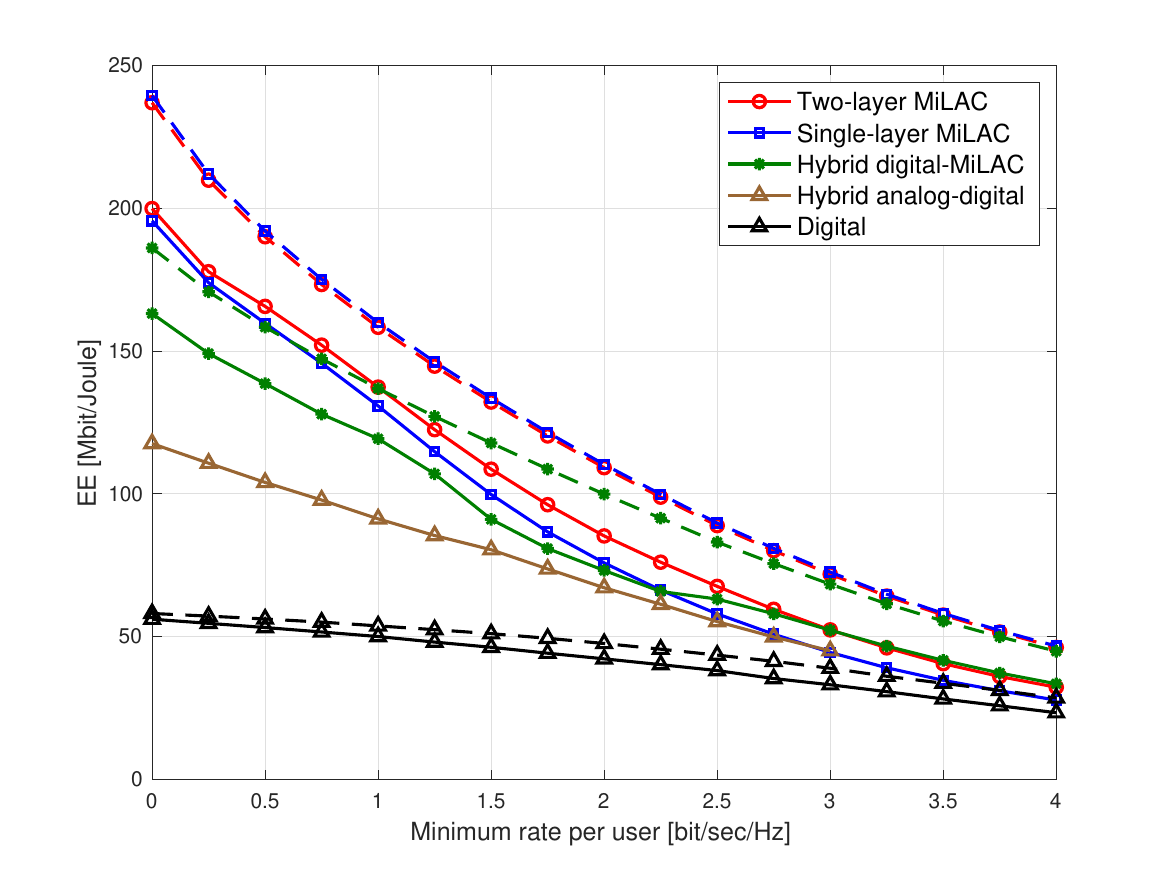}
         \vspace{-0.8cm}
        \caption{EE achieved by the proposed algorithm (solid) and the low-dimensional closed-form search (dashed) as the rate threshold increases.}
        \label{fig:vary_r1}
    \end{minipage}%
    \hfill
    \begin{minipage}{.32\textwidth}
        \centering
        \includegraphics[width=\columnwidth]{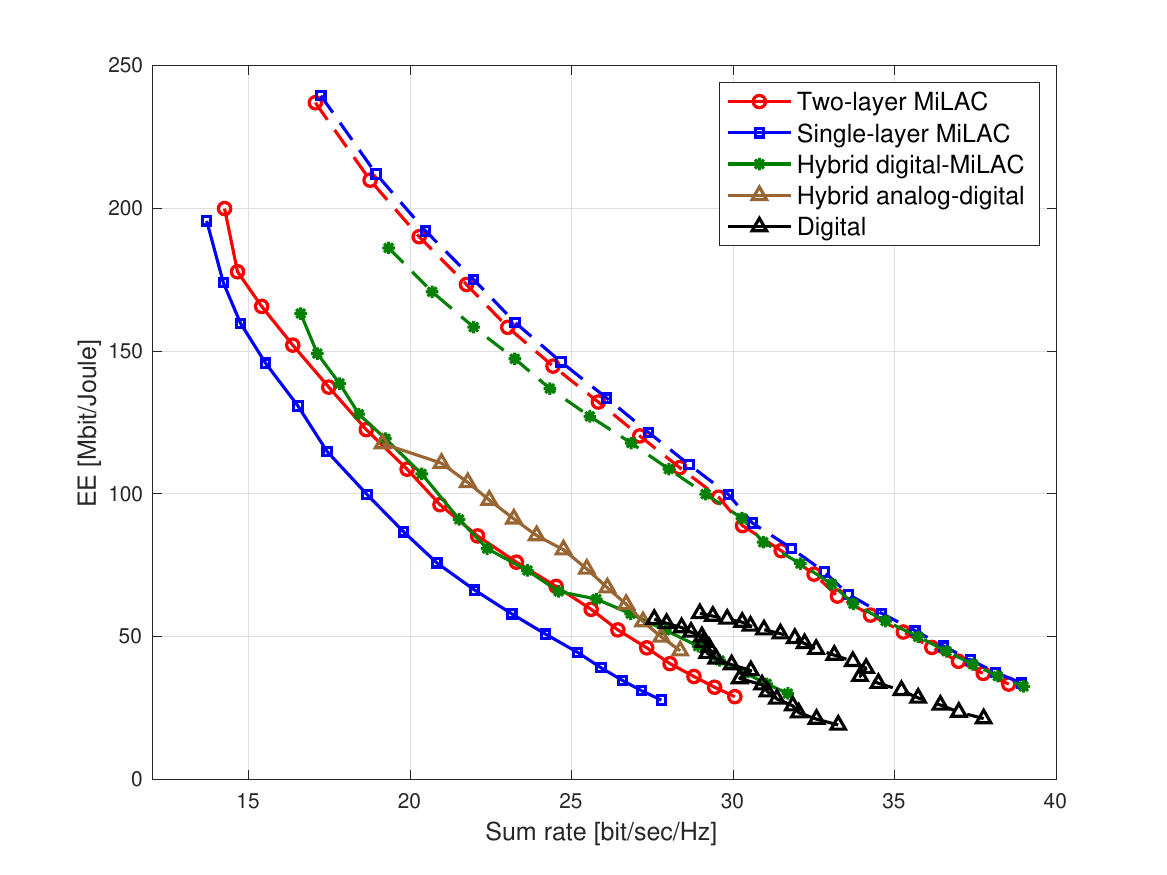}
         \vspace{-0.8cm}
        \caption{EE versus the sum rate achieved by the proposed algorithm (solid) and the low-dimensional closed-form search (dashed) as the rate threshold increases.}
        \label{fig:vary_r3}
    \end{minipage}
    \vspace{-0.7cm}
\end{figure*}

We first plot the EE values of the MiLAC-based schemes (SLM, TLM, and HDM) and conventional beamforming schemes (DBF and HAD) at each iteration of the proposed algorithm in Fig.~\ref{fig:conv}, along with the corresponding EE values obtained from the low-dimensional closed-form search algorithm presented in Section~\ref{sec:Asy_ana}. It can be observed that all algorithms converge within a few iterations (fewer than eight), demonstrating the fast convergence of the proposed optimization framework.

At convergence, the TLM, SLM, and HDM systems achieve EE improvements of 245\%, 231\%, and 185\%, respectively, compared with the DBF system, and 67\%, 60\%, and 38\%, respectively, compared with the HAD system. These improvements stem primarily from the fact that the TLM and SLM architectures reduce the number of RF chains from $N=64$ in the DBF system and $M=12$ in the HAD system to only $K=8$. Furthermore, although the HDM architecture still employs $M=12$ RF chains, it eliminates the need for a dedicated phase shifter at each antenna by replacing the phase-shifting network with a tunable impedance network, which consumes significantly less circuit power.

Finally, although the channel realizations considered in these simulations are not orthogonal, the low-dimensional closed-form search algorithm presented in Section~\ref{sec:Asy_ana} still provides a useful low-complexity approximation of the achievable EE without requiring the optimization problem to be solved until convergence. This is particularly useful for evaluating different system configurations and design parameters, enabling efficient performance assessment without repeatedly executing the iterative optimization algorithm.
\vspace{-0.3cm}
\subsection{Impact of Varying the Minimum Rate Constraint}

Fig.~\ref{fig:vary_r1} depicts the impact of increasing the minimum rate threshold on the EE of the considered systems. As shown in the figure, for low values of the rate threshold, the MiLAC-based systems achieve significantly higher EE compared with the DBF and HAD systems, with the TLM achieving a slight advantage over the SLM and HDM systems. As the rate threshold increases, however, the EE of all five systems degrades due to the increased rate requirements, and the performance gap between the different systems becomes smaller. It should also be noted that the optimization problem becomes infeasible for the HAD system (which employs a lower number of RF chains compared with the DBF system) when the rate threshold exceeds 3~bit/sec/Hz. At high rate thresholds (e.g., 4~bit/sec/Hz), the TLM, SLM, and HDM systems still achieve EE improvements of approximately 1.2--1.45$\times$ compared with the DBF system.

Fig.~\ref{fig:vary_r3} illustrates the EE versus the sum rate of the considered systems as the minimum rate threshold increases. It can be observed that DBF achieves higher sum rates due to its fully digital beamforming capability. However, for a given rate requirement, DBF exhibits lower EE than the MiLAC-based architectures because of its significantly higher circuit power consumption. As the minimum rate requirement increases, the EE gap between DBF and MiLAC gradually reduces. This is because the relative impact of the circuit power consumption becomes less significant at higher transmission rates, where the transmit power required to satisfy the rate constraint becomes the dominant component of the total power consumption. The different EE degradation trends observed in Fig.~\ref{fig:vary_r1} can be explained by Fig.~\ref{fig:vary_r3}. Specifically, the EE of the DBF system decreases by approximately 59\% as the rate threshold increases from 0 to 4~bit/sec/Hz, compared with reductions of 83.5\%, 86\%, and 80\% for the TLM, SLM, and HDM systems, respectively. This behavior can be attributed to the substantially higher circuit power consumption of the DBF architecture. For a given sum-rate requirement, the higher circuit power results in lower EE compared with the MiLAC-based architectures. As the minimum rate threshold increases, however, the EE gap between DBF and MiLAC-based architectures becomes smaller. This is because the increasing rate requirement leads to higher transmit power consumption, making the rate-dependent power component increasingly significant relative to the difference in circuit power consumption. Consequently, although DBF remains less energy efficient, its relative EE disadvantage diminishes at higher rate requirements.

\vspace{-0.2cm}

\subsection{Effect of Varying the Number of Transmit Antennas}
\begin{figure*}[t]
    \centering
    
    \begin{minipage}{0.32\textwidth}
        \centering
        \includegraphics[width=\columnwidth]{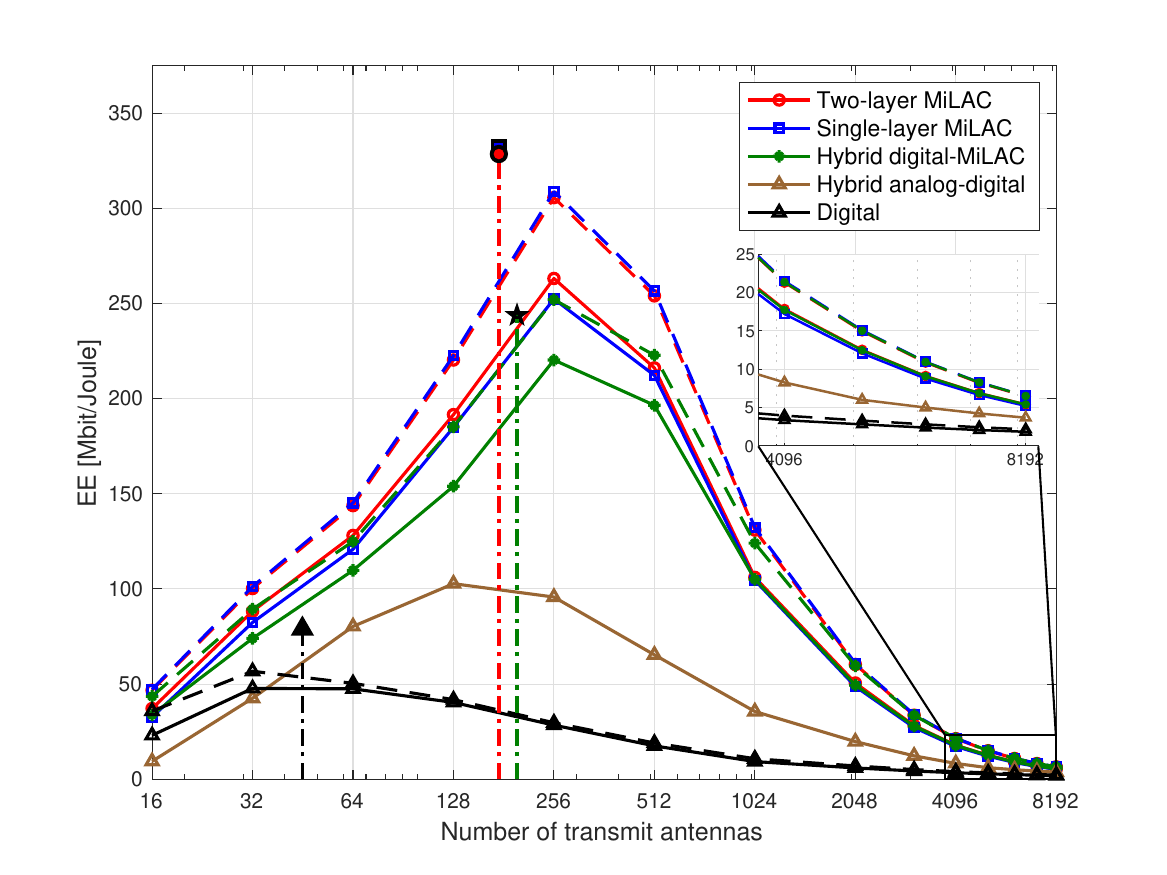}
        \vspace{-0.8cm}
        \caption{EE achieved by the proposed algorithm (solid) and the low-dimensional closed-form search (dashed) as the number of antennas increases. Stem plots indicate the approximate optimal $N$ and EE values from Lemma~\ref{lemma:2}.}
        \label{fig:vary_N1}
    \end{minipage}
    \hfill
    \begin{minipage}{0.32\textwidth}
        \centering
        \includegraphics[width=\columnwidth]{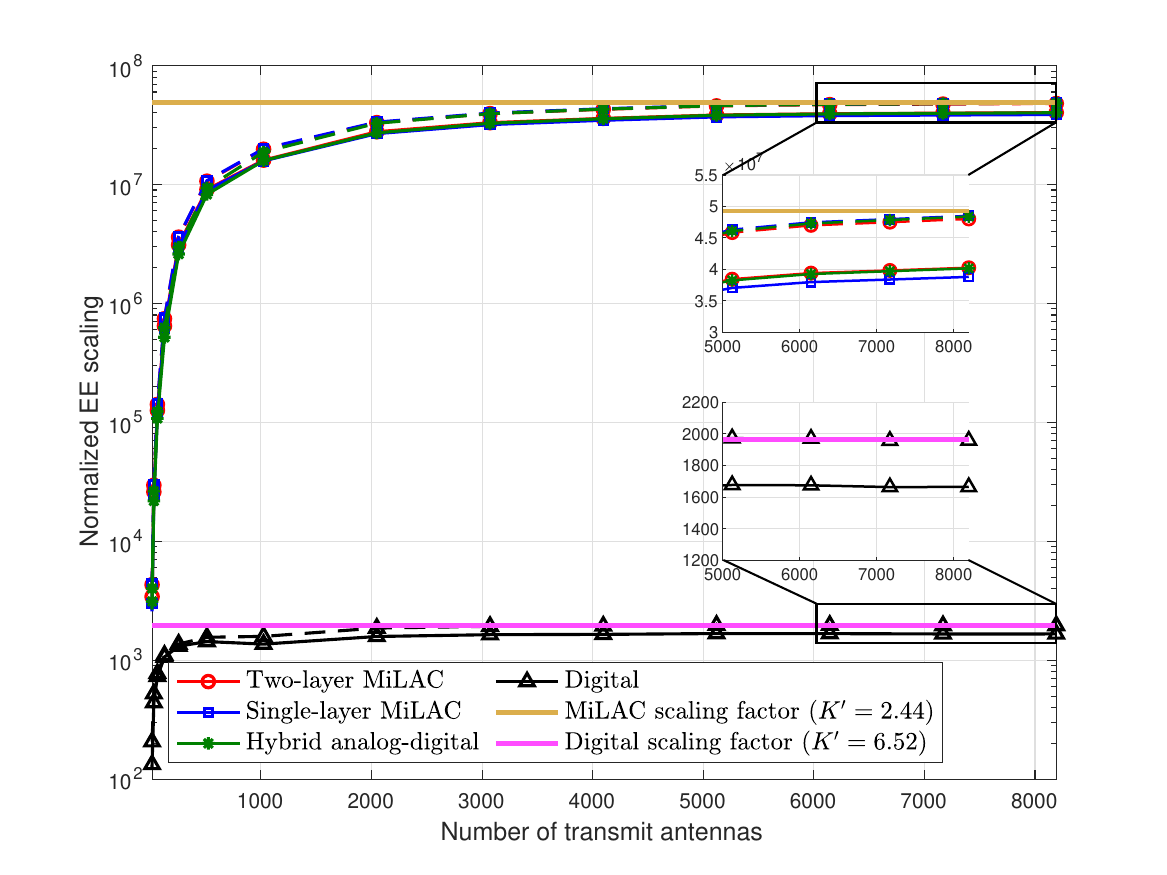}
                \vspace{-0.8cm}
        \caption{Normalized EE versus the number of antennas achieved by the proposed algorithm (solid) and the low-dimensional closed-form search (dashed) as the number of transmit antennas increases. Horizontal lines indicate the corresponding asymptotic scaling constants.}
        \label{fig:vary_N2}
    \end{minipage}
    \hfill
    \begin{minipage}{0.32\textwidth}
        \centering
        \includegraphics[width=\columnwidth]{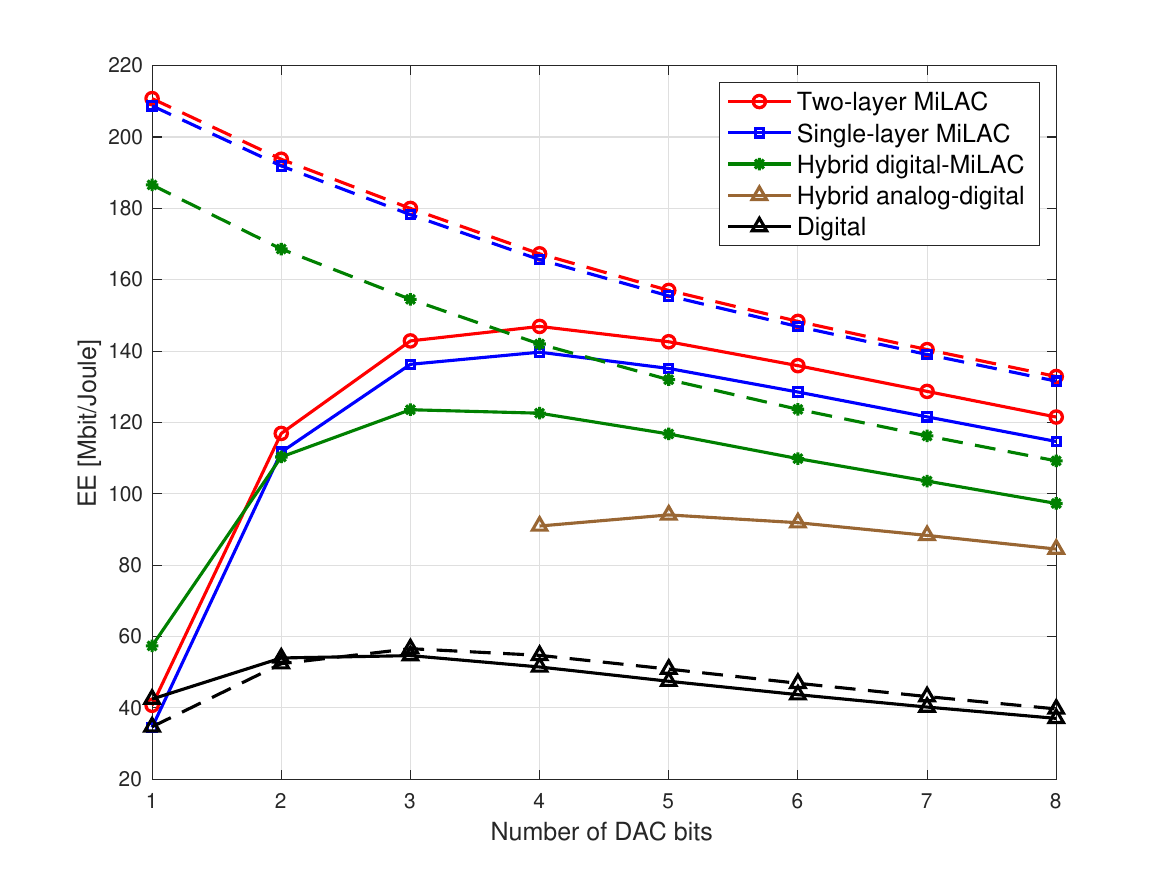}
                        \vspace{-0.8cm}
        \caption{EE achieved by the proposed algorithm (solid) and the low-dimensional closed-form search (dashed) as the number of DAC bits increases.}
        \label{fig:vary_b}
    \end{minipage}
\vspace{-0.75cm}
\end{figure*}

Fig.~\ref{fig:vary_N1} illustrates the impact of the number of transmit antennas $N$ on the EE of the considered architectures. It can be observed that DBF and HAD achieve their maximum EE at relatively small array sizes, specifically between $N=32$ and $64$ for DBF and between $N=128$ and $256$ for HAD. In contrast, the MiLAC-based architectures attain their maximum EE at larger array sizes, between $N=256$ and $512$. Although the circuit power consumption of the MiLAC-based architectures increases quadratically with $N$ due to the impedance-tuning network, their overall hardware power consumption remains substantially lower than that of the RF components required by DBF and HAD over small-to-moderate array sizes. Consequently, increasing $N$ initially provides sufficient sum-rate gains to compensate for the additional circuit power in the MiLAC-based architectures, shifting their EE-maximizing array size to larger $N$.

The stem plots in Fig.~\ref{fig:vary_N1} represent the approximate optimal antenna number $N^{\mathrm{opt}}$ and the corresponding peak EE approximated by Lemma~\ref{lemma:2}. The approximate values of $N^{\mathrm{opt}}$ closely track the array size at which the simulated EE curves actually peak for all four architectures, and the corresponding EE values from~\eqref{eq:eta_at_peak_full} likewise align well with the simulated peaks, despite Lemma~\ref{lemma:2} being derived under the idealized orthogonal-channel and negligible-quantization-noise assumptions of Section~\ref{sec:Asy_ana}. This confirms that, even though the closed-form EE search of Algorithm~\ref{alg:rate_margin_SLM} is derived for an idealized system, the resulting stationarity condition still captures the essential trade-off between beamforming gain and circuit power that governs the true EE-optimal array size, making it a useful low-complexity predictor of $N^{\mathrm{opt}}$.

As $N$ becomes very large, the asymptotic EE of the MiLAC-based architectures decays faster than that of DBF, scaling as $\ln(N)/N^2$ compared with $\ln(N)/N$ for DBF, as derived in Section~\ref{sec:Asy_ana}. To validate these theoretical scaling orders, Fig.~\ref{fig:vary_N2} plots the corresponding normalized EE metrics as a function of $N$. Specifically, the EE is normalized by its respective theoretical decay rate, i.e., $\eta^g N^2/\ln(N)$ for $g\in\{\mathrm{SLM},\mathrm{TLM},\mathrm{HDM}\}$ and $\eta^g N/\ln(N)$ for $g=\mathrm{DBF}$, where $\eta^g$ denotes the EE of architecture $g$. This normalization removes the dominant dependence on $N$ predicted by the asymptotic analysis. Therefore, if the derived scaling laws accurately describe the large-antenna behavior, the normalized EE should converge to the corresponding asymptotic scaling factor.

As depicted in Fig.~\ref{fig:vary_N2}, these normalized metrics gradually converge to constant asymptotic factors, indicated by the horizontal lines in Fig.~\ref{fig:vary_N2}. As derived in Section~\ref{sec:Asy_ana}, these factors are $\frac{6K'}{P_\mathrm{IT}\ln 2}$ for the SLM, TLM, and HDM architectures and $\frac{2K'}{P_\mathrm{RF}}$ for DBF. Here, $K'$ denotes the average number of users with non-zero rate margins at the asymptotic EE optimum, with $K'=2.44$ for the MiLAC-based architectures and $K'=6.52$ for DBF\footnote{For the Monte Carlo results, the asymptotic scaling factors are evaluated using the empirical average of $K'$ across channel realizations.}. The larger $K'$ observed for DBF is consistent with its substantially higher hardware power consumption, particularly as $N$ increases. In this case, the EE-optimal solution tends to operate at a higher overall sum rate, with more users contributing non-zero rate margins. Thus, the larger $K'$ indicates that DBF exploits multi-user transmission more extensively at its asymptotic EE optimum. Moreover, the convergence of the MiLAC-based architectures to their asymptotic factors occurs at substantially larger $N$ than for DBF. This is because $P_\mathrm{IT}$ is relatively small compared with $P_\mathrm{RF}$, so the quadratic circuit-power term in the MiLAC architectures becomes dominant only at much larger antenna-array sizes. Consequently, the asymptotic scaling behavior of the MiLAC-based architectures becomes clearly visible only in the very large-antenna regime.

\vspace{-0.4cm}

\subsection{Effect of Varying the Number of DAC Bits}
Fig.~\ref{fig:vary_b} demonstrates the impact of varying the number of DAC bits on the EE of the considered systems. While increasing the number of bits reduces the quantization noise and improves the achievable performance, it also increases the power consumption of the ADC/DACs according to the model provided in Table~\ref{tab:2}. When the number of DAC bits is one, the EE of the HDM system outperforms those of the SLM and TLM systems by 50\% and 65\%, respectively, due to the digital beamforming stage in the HDM architecture, which is capable of preprocessing the quantization noise. However, as the number of DAC bits increases, the SLM and TLM architectures achieve better EE performance due to their reduced hardware power consumption. Furthermore, since the EE values obtained from the low-dimensional closed-form search are derived under the assumption of negligible quantization noise, they provide inaccurate approximations of the actual system EE for low-resolution DACs. It can also be observed that the HAD architecture requires at least four DAC bits to satisfy the rate constraints and render the optimization problem feasible.  

When the number of DAC bits exceeds four, however, the quantization noise becomes sufficiently small, and further increasing the DAC resolution provides only very marginal performance improvements. Consequently, due to the additional power consumption introduced by higher-resolution ADC/DACs, the system EE degrades. As a result, a DAC resolution should be selected to balance the reduction of quantization noise and the additional ADC/DAC power consumption.
\vspace{-0.4cm}
\subsection{Effect of Varying the Transmit Power Budget}

\begin{figure}
         \centering 
         \includegraphics[width=0.65\columnwidth]{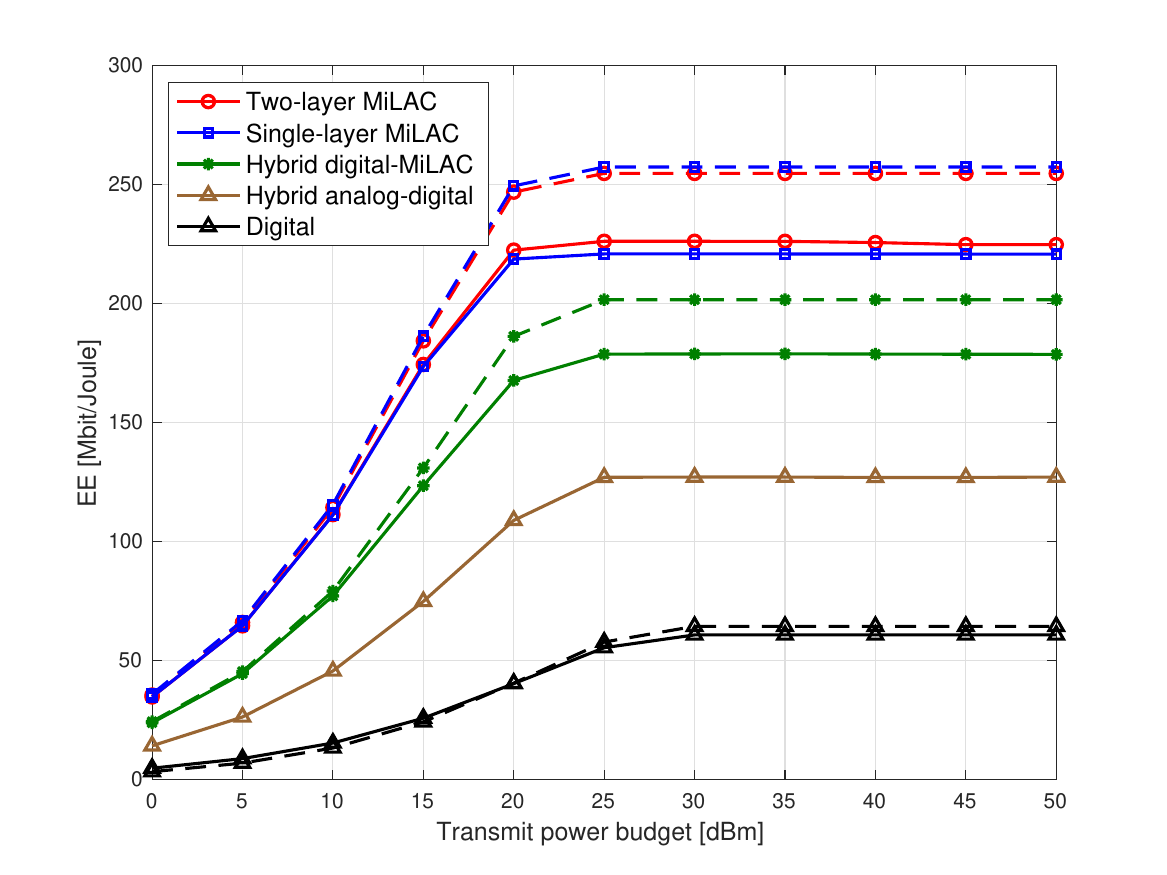}
                         \vspace{-0.5cm}
        \caption{EE achieved by the proposed algorithm (solid) and the low-dimensional closed-form search (dashed) as the transmit power budget increases.}
        \label{fig:vary_P}
\end{figure}
Fig.~\ref{fig:vary_P} depicts the impact of varying the transmit power budget on the EE of the considered systems. For this case only, the rate threshold is set to zero (i.e., \(r=0\)) to avoid infeasible optimization problems at low power budgets and to provide a complete characterization of the system behavior. It can be observed that increasing the power budget initially improves the EE of all considered systems by enabling higher achievable rates. However, when the power budget exceeds \unit[30]{dBm}, the systems operate below the available power limit, as further increasing the transmit power leads to a faster increase in power consumption than the corresponding rate improvement, thereby degrading the EE. Moreover, for power budgets up to \unit[20]{dBm}, the EE of the SLM and TLM architectures remains very close, with noticeable differences emerging only at higher power budgets. This is because, in the low-power regime, both architectures are primarily limited by the available transmit power, and their different circuit power consumptions have a relatively small impact on the overall EE. As the power budget increases, the higher beamforming flexibility of the TLM architecture allows it to better exploit the available transmit power, resulting in improved EE compared with the SLM architecture.
\vspace{-0.5cm}
\subsection{Effect of the Number of Users}
\begin{figure}
         \centering 
         \includegraphics[width=0.65\columnwidth]{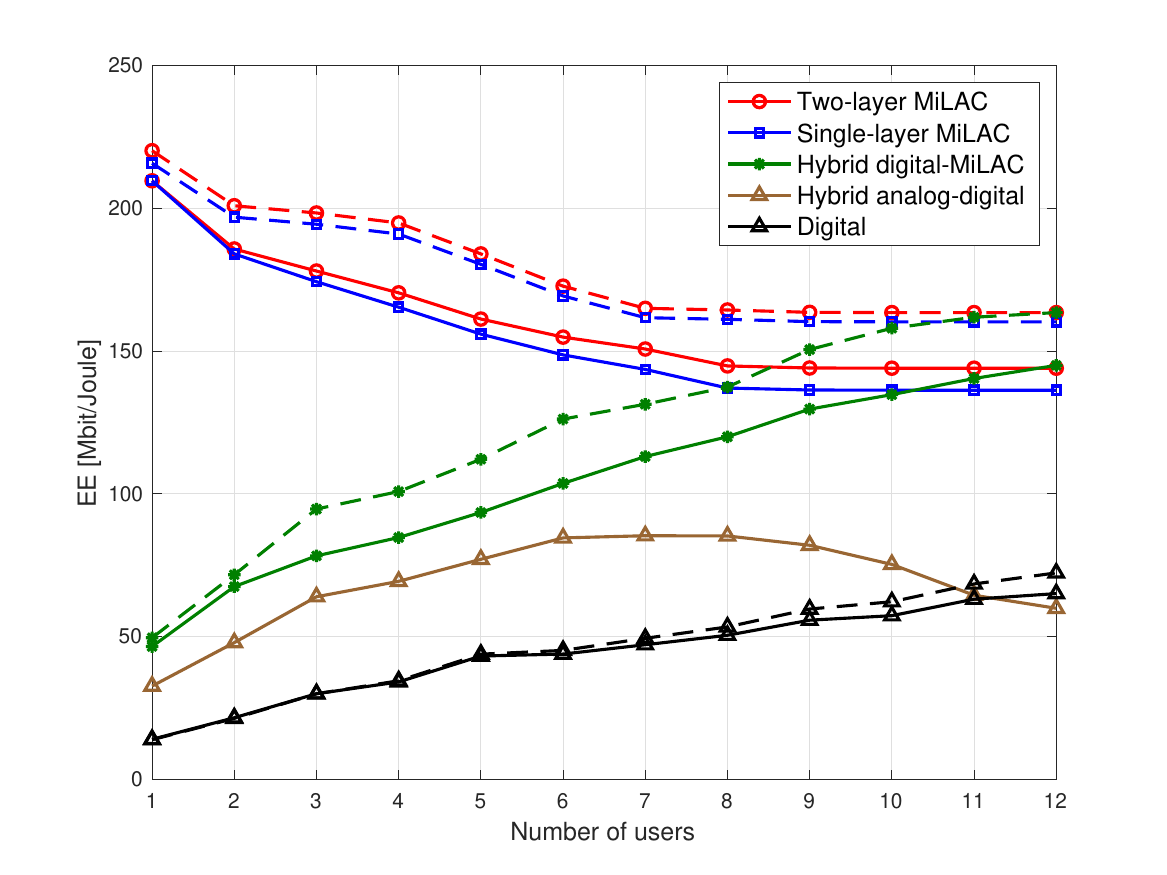}
                         \vspace{-0.5cm}
        \caption{EE achieved by the proposed algorithm (solid) and the low-dimensional closed-form search (dashed) as the number of users increases.}
        \label{fig:vary_K}
\end{figure}

Fig.~\ref{fig:vary_K} depicts the impact of the number of users on the EE of the considered systems. Interestingly, different system architectures exhibit different trends. For the TLM and SLM systems, increasing the number of users requires additional RF chains, as the number of RF chains in these architectures is equal to the number of users. This results in increased power consumption, whose growth rate exceeds the gain obtained from multiuser diversity, thereby reducing the EE. On the other hand, the number of RF chains in the HDM architecture is fixed to 12 in our simulations. Therefore, serving more users does not increase the RF-chain power consumption and can improve the sum rate by exploiting multiuser diversity, which in turn enhances the system EE. When the number of users reaches 12, the EE of the TLM and HDM architectures converges to a similar value. A similar trend is also observed for the DBF system, where the number of RF chains is equal to the number of transmit antennas.

The HAD system follows another interesting trend. Since the number of RF chains is also fixed to 12, the power consumption does not increase as the system serves more users. This, together with the exploitation of multiuser diversity, improves the EE up to $K=7$ users, after which the EE starts to decrease. This is because all users must satisfy the minimum rate constraint, and, given the limited degrees of freedom available in this system compared with the other architectures, additional transmit power is required to support the increased number of users, resulting in higher power consumption and reduced EE.
\vspace{-0.4cm}
\section{Conclusion} \label{sec:conc}
\vspace{-0.2cm}
In this paper, we investigated the EE of three MiLAC-based architectures, namely, SLM, TLM, and HDM. We formulated their EE maximization problems and developed reduced-dimensional SCA-based solutions, together with a low-dimensional closed-form search approach for efficient EE evaluation. We further characterized the asymptotic EE behavior, showing that MiLAC-based architectures scale as $\ln(N)/N^2$, whereas DBF scales as $\ln(N)/N$ under the considered power consumption models. The numerical results demonstrate that MiLAC-based architectures can substantially outperform conventional DBF and HAD over practically relevant antenna-array sizes. Among the considered architectures, TLM generally achieves the highest EE, while HDM becomes more advantageous with low-resolution RF chains due to its ability to mitigate quantization noise. The EE advantage of MiLAC architectures also depends on the number of users and antennas. These results highlight the potential of MiLAC architectures for energy-efficient large-scale MIMO systems and identify the regimes in which each architecture is most beneficial.

\vspace{-0.4cm}

\ifCLASSOPTIONcaptionsoff
  \newpage
\fi





\appendices

\bibliographystyle{IEEEtran}
\bibliography{IEEEabrv,Bibliography}

\begin{thebibliography}{10}
\providecommand{\url}[1]{#1}
\csname url@rmstyle\endcsname
\providecommand{\newblock}{\relax}
\providecommand{\bibinfo}[2]{#2}
\providecommand\BIBentrySTDinterwordspacing{\spaceskip=0pt\relax}
\providecommand\BIBentryALTinterwordstretchfactor{4}
\providecommand\BIBentryALTinterwordspacing{\spaceskip=\fontdimen2\font plus
\BIBentryALTinterwordstretchfactor\fontdimen3\font minus \fontdimen4\font\relax}
\providecommand\BIBforeignlanguage[2]{{%
\expandafter\ifx\csname l@#1\endcsname\relax
\typeout{** WARNING: IEEEtran.bst: No hyphenation pattern has been}%
\typeout{** loaded for the language `#1'. Using the pattern for}%
\typeout{** the default language instead.}%
\else
\language=\csname l@#1\endcsname
\fi
#2}}

\bibitem{2015_Bjornson}
E.~Bj\"ornson, \emph{et~al.}, ``Optimal design of energy-efficient multi-user {MIMO} systems: Is massive {MIMO} the answer?'' \emph{IEEE Trans. Wireless Commun.}, vol.~14, no.~6, pp. 3059--3075, 2015.

\bibitem{2025_You}
C.~You \emph{et~al.}, ``Next generation advanced transceiver technologies for {6G} and beyond,'' \emph{IEEE J. Sel. Areas Commun.}, vol.~43, no.~3, pp. 582--627, 2025.

\bibitem{2016_Sohrabi}
F.~Sohrabi and W.~Yu, ``Hybrid digital and analog beamforming design for large-scale antenna arrays,'' \emph{IEEE J. Sel. Top. Signal Process.}, vol.~10, no.~3, pp. 501--513, 2016.

\bibitem{2017_Abbas}
W.~B. Abbas, \emph{et~al.}, ``Millimeter wave receiver efficiency: A comprehensive comparison of beamforming schemes with low resolution {ADCs},'' \emph{IEEE Trans. Wireless Commun.}, vol.~16, no.~12, pp. 8131--8146, 2017.

\bibitem{2025_Magbool1}
A.~Magbool, \emph{et~al.}, ``A survey on integrated sensing and communication with intelligent metasurfaces: Trends, challenges, and opportunities,'' \emph{IEEE Open J. Commun. Soc.}, vol.~6, pp. 7270--7318, 2025.

\bibitem{2026_Magbool}
------, ``Beyond the limits of rigid arrays: Flexible intelligent metasurfaces for next-generation wireless networks,'' \emph{arXiv preprint arXiv:2603.11886}, 2026.

\bibitem{2026_Zhu}
L.~Zhu \emph{et~al.}, ``A tutorial on movable antennas for wireless networks,'' \emph{IEEE Commun. Surv. Tutor.}, vol.~28, pp. 3002--3054, 2026.

\bibitem{2025_Nerini}
M.~Nerini and B.~Clerckx, ``Analog computing for signal processing and communications – part {I}: Computing with microwave networks,'' \emph{IEEE Trans. Signal Process.}, vol.~73, pp. 5183--5197, 2025.

\bibitem{2025_Nerini1}
------, ``Analog computing for signal processing and communications – part {II}: Toward gigantic {MIMO} beamforming,'' \emph{IEEE Trans. Signal Process.}, vol.~73, pp. 5198--5212, 2025.

\bibitem{2026_Nerini2}
------, ``Physics-compliant modeling and optimization of {MIMO} systems aided by microwave linear analog computers,'' \emph{arXiv preprint arXiv:2602.19379}, 2026.

\bibitem{2026_Nerini3}
------, ``Capacity of {MIMO} systems aided by microwave linear analog computers ({MiLACs}),'' \emph{arXiv preprint arXiv:2506.05983}, 2025.

\bibitem{2026_Wu}
Z.~Wu, \emph{et~al.}, ``Microwave linear analog computer ({MiLAC})-aided multiuser {MISO}: Fundamental limits and beamforming design,'' \emph{arXiv preprint arXiv:2601.10060}, 2026.

\bibitem{2026_Nerini4}
M.~Nerini and B.~Clerckx, ``{MIMO} systems aided by microwave linear analog computers: Capacity-achieving architectures with reduced circuit complexity,'' \emph{IEEE Trans. Wireless Commun.}, vol.~25, pp. 14\,597--14\,610, 2026.

\bibitem{2026_Zhang1}
Y.~Zhang, \emph{et~al.}, ``Beamforming design for stem-connected microwave linear analog computer ({MiLAC})-aided multiuser {MISO} downlinks,'' \emph{arXiv preprint arXiv:2606.14499}, 2026.

\bibitem{2026_Nerini5}
M.~Nerini and B.~Clerckx, ``Microwave linear analog computer ({MiLAC}) for simultaneous active and passive beamforming,'' \emph{arXiv preprint arXiv:2605.31549}, 2026.

\bibitem{2026_Peng}
Y.~Peng, \emph{et~al.}, ``Hybrid digital and microwave linear analog computer ({MiLAC})-aided beamforming for multiuser {MIMO-OFDM} systems,'' \emph{arXiv preprint arXiv:2604.26532}, 2026.

\bibitem{2026_Liu}
Z.~Liu, \emph{et~al.}, ``Microwave linear analog computer {MiLAC}-aided {MIMO} radar sensing: Transmit beamforming design and {DoA} estimation,'' \emph{arXiv preprint arXiv:2605.21020}, 2026.

\bibitem{2026_Zhang2}
Y.~Zhang, \emph{et~al.}, ``How many {RF} chains does a microwave linear analog computer ({MiLAC}) need to match the fully-digital cram{\'e}r--rao bound?'' \emph{arXiv preprint arXiv:2606.23986}, 2026.

\bibitem{2026_Zhou}
X.~Zhou, \emph{et~al.}, ``Two-layer microwave linear analog computer ({MiLAC})-aided multi-user {MISO} networks,'' \emph{arXiv preprint arXiv:2604.24303}, 2026.

\bibitem{2026_Zhang}
Y.~Zhang, \emph{et~al.}, ``Quantization-aware {EE} optimization and {SE-EE} tradeoff for {MiLAC}-aided {MU-MISO} beamforming,'' \emph{arXiv preprint arXiv:2604.24538}, 2026.

\bibitem{2025_Magbool}
A.~Magbool, \emph{et~al.}, ``Robust beamforming design for fairness-aware energy efficiency maximization in {RIS}-assisted {mmWave} communications,'' \emph{IEEE Trans. Commun.}, vol.~73, no.~4, pp. 2648--2662, 2025.

\bibitem{2019_Dai}
J.~Dai, \emph{et~al.}, ``Achievable rates for full-duplex massive {MIMO} systems with low-resolution {ADCs/DACs},'' \emph{IEEE Access}, vol.~7, pp. 24\,343--24\,353, 2019.

\bibitem{2016_Buzzi}
S.~Buzzi, \emph{et~al.}, ``A survey of energy-efficient techniques for {5G} networks and challenges ahead,'' \emph{IEEE J. Sel. Areas Commun.}, vol.~34, no.~4, pp. 697--709, 2016.

\bibitem{2018_Ribeiro}
L.~N. Ribeiro, \emph{et~al.}, ``Energy efficiency of {mmWave} massive {MIMO} precoding with low-resolution {DACs},'' \emph{IEEE J. Sel. Topics Signal Process.}, vol.~12, no.~2, pp. 298--312, 2018.

\bibitem{2017_Roth}
K.~Roth and J.~A. Nossek, ``Achievable rate and energy efficiency of hybrid and digital beamforming receivers with low resolution {ADC},'' \emph{IEEE J. Sel. Areas Commun.}, vol.~35, no.~9, pp. 2056--2068, 2017.

\bibitem{2014_Grant}
M.~Grant and S.~Boyd, ``{CVX}: Matlab software for disciplined convex programming, version 2.1,'' \url{https://cvxr.com/cvx}, Mar. 2014.

\bibitem{2014_Rusek}
F.~Rusek, \emph{et~al.}, ``Scaling up {MIMO}: Opportunities and challenges with very large arrays,'' \emph{EEE Signal Process. Mag.}, vol.~30, no.~1, pp. 40--60, 2013.

\bibitem{2016_Marzetta}
T.~L. Marzetta, \emph{et~al.}, \emph{Fundamentals of Massive {MIMO}}.\hskip 1em plus 0.5em minus 0.4em\relax Cambridge University Press, 2016.

\bibitem{2013_Ngo}
H.~Q. Ngo, \emph{et~al.}, ``Energy and spectral efficiency of very large multiuser {MIMO} systems,'' \emph{IEEE Trans. Commun.}, vol.~61, no.~4, pp. 1436--1449, 2013.

\bibitem{1996_Corless}
R.~M. Corless, \emph{et~al.}, ``On the lambert {W} function,'' \emph{Advances in Computational Mathematics}, vol.~5, no.~1, pp. 329--359, 1996.

\bibitem{2012_Isheden}
C.~Isheden, \emph{et~al.}, ``Framework for link-level energy efficiency optimization with informed transmitter,'' \emph{IEEE Trans. Wireless Commun.}, vol.~11, no.~8, pp. 2946--2957, 2012.

\bibitem{2014_Huang}
Y.~Huang and L.~Qiu, ``On the energy efficiency–spectral efficiency tradeoff in random beamforming,'' \emph{IEEE Wirel. Commun. Lett.}, vol.~3, no.~5, pp. 461--464, 2014.

\end{thebibliography}

%


\end{document}